\documentclass[12pt]{article}
\usepackage{latexsym, amsmath, amsfonts, amsthm, amssymb}
\usepackage{times}
\usepackage{a4wide}
\usepackage{times}
\usepackage{graphicx, tocloft}
\usepackage{mathrsfs}
\usepackage{tikz}
\usetikzlibrary{matrix,arrows}
\usepackage[title]{appendix}

\def \D {\, \, || \, \,}

\def \RR {\mathbb R}

\def \EE {\mathbb E}

\def \eps {\varepsilon}
\def \vphi {\varphi}

\def \cF {\mathcal F}

\newcommand{\m}{\mathcal}
\newcommand{\mmse}{\mathrm{mmse}}
\newcommand{\Var}{\mathrm{Var}}
\newcommand{\Cov}{\mathrm{Cov}}
\newcommand{\CLS}{C_{\mathrm{LS}}}
\newcommand{\CP}{C_{\mathrm{P}}}
\newcommand{\nCP}{\bar{C}_{\mathrm{P}}}
\newcommand{\nFisher}{\bar{\m{J}}}
\newcommand{\nEP}{\bar{N}}

\DeclareMathOperator*{\argmax}{\arg\!\max}

\newtheorem{theorem}{Theorem}[section]

\newtheorem{question}[theorem]{Question}
\newtheorem{lemma}[theorem]{Lemma}
\newtheorem{conjecture}[theorem]{Conjecture}

\newtheorem{proposition}[theorem]{Proposition}
\newtheorem{corollary}[theorem]{Corollary}

\theoremstyle{definition}
\newtheorem{remark}[theorem]{Remark}

\def\myffrac#1#2 in #3{\raise 2.6pt\hbox{$#3 #1$}\mkern-1.5mu\raise 0.8pt\hbox{$
		#3/$}\mkern-1.1mu\lower 1.5pt\hbox{$#3 #2$}}

\def\qed{\hfill $\vcenter{\hrule height .3mm
		\hbox {\vrule width .3mm height 2.1mm \kern 2mm \vrule width .3mm
			height 2.1mm} \hrule height .3mm}$ \bigskip}

\def \id {{\rm Id}}

\def \id {{ \rm Id}}

\def \susbeteq {\subseteq}

\begin{document}

\title{The strong data processing inequality under the heat flow} 
\author{B. Klartag, O. Ordentlich
\thanks{Bo'az Klartag  is with the Department of Mathematics, Weizmann Institute of Science, Rehovot 76100, Israel (\texttt{boaz.klartag@weizmann.ac.il}). Or Ordentlich is with the 
Hebrew University of Jerusalem, Israel (\texttt{or.ordentlich@mail.huji.ac.il}). The work of BK was supported by the Israel Science 
Foundation (ISF),
grant No. 765/19. The work of OO was supported by the Israel Science 
Foundation (ISF),
grant No. 1641/21.}}

\date{}
\maketitle

\begin{abstract}%
Let $\nu$ and $\mu$ be probability distributions on $\RR^n$, and $\nu_s,\mu_s$ be their evolution under the heat flow, that is, the probability distributions resulting from convolving their density with the density of an isotropic Gaussian random vector with variance $s$ in each entry. This paper studies the rate of decay of $s\mapsto D(\nu_s\|\mu_s)$ for various divergences, including the $\chi^2$ and Kullback-Leibler (KL) divergences. We prove upper and lower bounds on the strong data-processing inequality (SDPI) coefficients corresponding to the source $\mu$ and the Gaussian channel.
We also prove generalizations of de Brujin's identity, and Costa's result on the concavity in $s$ of the differential entropy of $\nu_s$. As a byproduct of our analysis, we obtain new lower bounds on the mutual information between $X$ and $Y=X+\sqrt{s} Z$, where $Z$ is a standard Gaussian vector in $\RR^n$, independent of $X$, and on the minimum mean-square error (MMSE) in estimating $X$ from $Y$, in terms of the Poincar\'e constant of $X$.
\end{abstract}

\section{Introduction}

For two probability distributions $\nu$ and $\mu$ on $\RR^n$, where $\nu$ is absolutely continuous with respect to $\mu$, and a smooth, convex function $\varphi:(0,\infty)\to\RR$ with $\varphi(1)=0$, the $\varphi$-divergence ~\cite{csiszar1967information} is defined as
\begin{align}
D_\varphi(\nu\|\mu)=\EE_{\mu}\left[\varphi\left(\frac{d\nu}{d\mu}\right) \right],
\label{eq:DphiDef}
\end{align}
where $\varphi(0)=\lim_{t \rightarrow 0^+}\varphi(t)$. Let $Z\sim\m{N}(0,\id)$ be a standard Gaussian random variable in $\RR^n$. For $s\geq 0$, denote by $\nu_s$ the probability distribution of the random variable $Y=X+\sqrt{s}Z$ when $X\sim \nu$ is independent of $Z$. Similarly, write $\mu_s$ for the probability distribution of $Y$ when $X\sim \mu$ is independent of $Z$.  By the data-processing inequality~\cite{csiszar1967information,PW_book} for $\varphi$-divergences, we have that $$ s\mapsto D_\varphi(\nu_s\|\mu_s) $$ is non-increasing in $s \geq 0$. The goal of this paper is to provide estimates on the rate of decay of $s\mapsto D_\varphi(\nu_s\|\mu_s)$ as a function of $\mu$ and uniformly over the measure $\nu$. We stress that here {\it both} measures evolve according to the heat flow. This is in contrast to prior work that analyzed the evolution of the entropy and similar functionals of $\nu_s$ according to the heat flow. This corresponds to taking $\mu$ as the Lebesgue measure whose heat flow evolution satisfies  $\mu_s=\mu$. See more details below. 

\medskip 
To that end, we study the \emph{strong data-processing inequality (SDPI) constant/contraction coefficient} for the probability distribution $\mu$ and the Gaussian channel ~\cite{ahlswede1976spreading},~\cite[chapter 33]{PW_book}, which is defined as
\begin{align}
 \eta_\varphi(\mu,s)\triangleq \sup_{\nu: 0<D_\varphi(\nu\|\mu)<\infty}\frac{D_\varphi(\nu_s\|\mu_s)}{D_\varphi(\nu\|\mu)}.   
\end{align}
Recalling that $D_\varphi(\nu_s\|\mu_s)\geq 0$ (with equality if and only if $\nu_s=\mu_s$), by the data-processing inequality we have that $0\leq \eta_\varphi(\mu,s)\leq 1$. We mostly analyze two choices of $\varphi$. The first is $\varphi(x)=\vphi_{\text{KL}}(x) = x\log x$,\footnote{In this paper, logarithms are always taken to the natural basis.} which results in the Kullback-Leibler (KL) divergence
\begin{align}
D(\nu\|\mu)=D_{\text{KL}}(\nu\|\mu)=\int_{\RR^n} d \nu  \log\frac{d\nu}{d\mu},
\end{align}
and the second is $\varphi(x)=\vphi_{\chi^2}(x) = (x-1)^2$, which results in the $\chi^2$-divergence
\begin{equation}
\chi^2(\nu\|\mu)=D_{\chi^2}(\nu\|\mu)=\int_{\RR^n} \left(\frac{d \nu}{d \mu} - 1 \right)^2 d \mu.
\label{eq_1116} 
\end{equation}
It is common to rewrite the integrand in (\ref{eq_1116}), by abuse of notation, 
 as $(d\nu-d \mu)^2 / d \mu$, 
where $d \mu$ refers to the ``density of $\mu$'' with respect to some ambient measure.
We denote the corresponding contraction coefficients by $\eta_{\text{KL}}(\mu,s)$ and $\eta_{\chi^2}(\mu,s)$. By~\cite[Theorem 2]{polyanskiy2017strong} we have that
\begin{equation}
\eta_{\chi^2}(\mu,s)\leq \eta_{\varphi}(\mu,s), \label{eq_1135}
\end{equation}
for any $\varphi$ with $\varphi''(1)>0$. In particular, $\eta_{\chi^2}(\mu,s)\leq \eta_{\text{KL}}(\mu,s)$.
Intuitively, inequality~\eqref{eq_1135} follows by considering measures $\nu$ that are very close to $\mu$, 
 and using Taylor approximation. 

\medskip
Some of our results hold for further choices of $\vphi$ beyond $\vphi_{\chi^2}$ and $\vphi_{\text{KL}}$. In particular, the crucial requirement for our results is that $1 / \vphi''$ is concave, as in Chafa\"i \cite{chafai}. This requirement is fulfilled for $\vphi_{\chi^2}$ and $\vphi_{\text{KL}}$ (in fact for those functions, $1/\vphi''$ is linear). Another important class of functions for which $1/\vphi''$ is concave is $\vphi_{\lambda}(x)=x^\lambda-1$, for $\lambda\in(1,2)$. The corresponding divergence $D_{\vphi_\lambda}(\nu\|\mu)$, defined by~\eqref{eq:DphiDef}, is related to the R\'enyi divergence of order $\lambda$, via the monotone transformation
\begin{align}
D_{\lambda}(\nu\|\mu)=\frac{1}{\lambda-1}\log (1+D_{\vphi_\lambda}(\nu\|\mu)).
\end{align}
Thus, although $D_{\lambda}(\nu\|\mu)$ is not a $\vphi$-divergence, it does satisfy the data processing inequality, and bounds on the decay rate of $s\mapsto D_{\vphi_\lambda}(\nu_s\|\mu_s)$ immediately imply bounds on the decay rate of the R\'enyi divergence $s\mapsto D_{\lambda}(\nu_s\|\mu_s)$. We note that the concavity requirement on $1/\vphi''$ fails to hold for many other important $\vphi$-divergences, such as the Jensen-Shannon divergence ($\vphi(x) = x \log \frac{2x}{x+1} + \log \frac{2}{x+1}$), 
the Le Cam divergence ($\vphi(x) = (1-x) / (2x + 2)$), 
the total variation divergence ($\vphi(x) = |x-1| / 2$) and the Hellinger divergence ($\vphi(x) = (1 - \sqrt{x})^2$).

\medskip
It is well-known,~see. e.g.,~\cite[Theorem 33.12]{PW_book} that
\begin{align}
\sqrt{\eta_{\chi^2}(\mu,s)}=S(X,X+\sqrt{s}Z)
\end{align}
where $X\sim\mu$ is independent of the standard Gaussian $Z$. Here $S(X,Y)$ is the {\it Hirschfeld-Gebelein-R\'enyi maximal correlation}
~\cite{renyi1959measures} between the random variables $X$ and $Y$ defined as
\begin{align}
	\label{eq_1159}
S(X,Y)&=\sup_{f,g:\RR^n\to\RR}\frac{\EE f(X)g(Y)-\EE(f(X))\EE(g(Y))}{\sqrt{\Var(f(X))\Var(g(Y))}}\\
&=\sup_f \sqrt{ \Var(\EE[f(X)|Y]) }.\label{eq:HGR_CS}
\end{align}
The supremum in \eqref{eq:HGR_CS} runs over all measurable functions $f:\RR^n\to \RR$ with $\EE f(X)=0$ and $\Var(f(X))=1$. In the passage from (\ref{eq_1159}) to \eqref{eq:HGR_CS}  we use the Cauchy-Schwartz inequality.

\medskip
The quantities $\eta_{\text{KL}}(\mu,s)$ and $\eta_{\chi^2}(\mu,s)$ can be thought of as (normalized) measures of statistical dependence between $X\sim \mu$ and $Y=X+\sqrt{s}Z$. Two other (un-normalized) measures of dependence we will consider are the minimum mean-square error (MMSE)\footnote{For a vector $z = (z_1,\ldots,z_n) \in\RR^n$ we denote by $|z|=\|z\|_2 = \sqrt{\sum_j z_j^2}$ the $\ell_2$ norm of $z$.}
\begin{align}
\mmse(\mu,s)= \EE|X-\EE[X|X+\sqrt{s}Z]|^2, 
\end{align}
and the mutual information
\begin{align*}
I(\mu,s)= I(X;X+\sqrt{s}Z)=h(X + \sqrt{s} Z) - h(X + \sqrt{s} Z | X) = h(\mu_s)-\frac{n}{2}\log(2\pi e s),    
\end{align*}
where for a probability distribution $\mu$ with density $\rho$ in $\RR^n$, the differential entropy of $X\sim\mu$ is 
\begin{align}
h(X)=h(\mu)=-\int_{\RR^n} \rho(x) \log \rho (x) dx.
\end{align}


\medskip Our results will be expressed in terms of the Poincar\'e constant $\CP(X)=\CP(\mu)$, of $X\sim\mu$, and the log-Sobolev constant $\CLS(X)=\CLS(\mu)$. The {\it Poincar\'e constant} of a random vector $X$ in $\RR^n$, denoted by $C_P(X)$, is the infimum over all $C \geq 0$ such that for any locally-Lipschitz function $f: \RR^n \rightarrow \RR$
satisfying $\EE |\nabla f(X)|^2 < \infty$,
\begin{equation} \Var( f(X) ) \leq C \cdot \EE |\nabla f(X)|^2. \label{eq_2154} \end{equation}
The {\it log-Sobolev constant} of a random vector $X$ in $\RR^n$, denoted by $\CLS(X)$, is the infimum over all $C \geq 0$ such that for any locally-Lipschitz function $f: \RR^n \rightarrow \RR$
satisfying $\EE f^2(X) = 1$ and $\EE |\nabla f(X)|^2 < \infty$,
\begin{equation} \EE f^2(X) \log f^2(X) \leq 2C \cdot \EE |\nabla f(X)|^2. \label{eq_1732} \end{equation}

We list a few important properties of the Poincar\'e and log-Sobolev constants, see  \cite{BGL} for more explanations, proofs and references.

\begin{itemize}
\item The Poincar\'e constant and the log-Sobolev constant of the standard Gaussian random vector $Z \in \RR^n$ satisfy $\CP(Z)=\CLS(Z)=1$. Furthermore, for any random variable we have that $$ \frac{\CP(X)}{\frac{1}{n}\EE|X-\EE X|^2}\geq 1 \qquad \text{and} \qquad \frac{\CLS(X)}{\frac{1}{n}\EE|X-\EE X|^2}\geq 1. $$ and in both cases the lower bound is attained if and only if $X$ is an isotropic Gaussian random variable.
\item $\CP(\mu)\leq \CLS(\mu)$.
\item For $p \geq 1$, let $\bar{\m{B}}_p=\{x\in\RR^n~:~\|x\|_p< \alpha\}$, where $\alpha$ is chosen such that $\mathrm{vol}(\bar{\m{B}}_p)=1$, be the unit-volume $\ell_p$-ball. For $U_p\sim\mathrm{Uniform}(\bar{\m{B}}_p)$, we have that
\begin{align}
c < \CP(U_p) < C 
\end{align}
for certain explicit universal constants $c, C > 0$. For $p \in [2, +\infty]$ we furthermore have
\begin{align}
c < \CLS(U_p) < C 
\end{align}
for universal constants $c, C > 0$. The log-Sobolev constant is not bounded by a universal constant when $p < 2$. 
These results may be extracted from the literature in the following way: First, up to a universal constant, Poincar\'e and log-Sobolev inequalities
follows from a corresponding isoperimetric inequality (see Ledoux \cite{led2} for the log-Sobolev case and Cheeger \cite{cheeger} for the Poincar\'e case). Second, the corresponding 
isoperimetric inequalities were proven in Sodin \cite{sodin} ($1 \leq p \leq 2$) and Lata\l{}a and Wojtaszczyk \cite{LW}.

\item $\CP(X)$ is finite for any log-concave random vector $X$ in $\RR^n$, see Bobkov \cite{bobkov}.
\item When $X$ is a log-concave random vector in $\RR^n$, the Kanann-Lov\'asz-Simonovits (KLS) conjecture from \cite{KLS} suggests that 
\begin{equation}
\| \Cov(X) \|_{op} \leq \CP(X) \leq C \cdot \| \Cov(X) \|_{op},
\label{eq_1557} 
\end{equation}
where $C > 0$ is a universal constant.
Here, $\Cov(X) \in \RR^{n \times n}$ is the covariance matrix of $X$, whose $i,j$ entry is $\EE X_i X_j - \EE X_i \EE X_j$.
We write $\| \cdot \|_{op}$ for the operator norm, thus\footnote{We denote the standard inner product of two vectors $x,y\in\RR^n$ by $x\cdot y=\langle x, y \rangle = x^T y$.}
\begin{align}  
\| \Cov(X) \|_{op} = \sup_{|\theta| = 1} \Var(X \cdot \theta) 
\label{eq_1558} 
\end{align}
where the supremum runs over all unit vectors $\theta \in \RR^n$. The left-hand side inequality in (\ref{eq_1557}) 
follows from the definitions (\ref{eq_2154}) and (\ref{eq_1558}). The current state of the art regarding 
the conjectural right-hand side inequality in (\ref{eq_1557}) is the bound from \cite{K}
\begin{align} 
\CP(X) \leq C \log n \cdot \| \Cov(X) \|_{op},
\end{align}
where $C > 0$ is a universal constant and $X$ is an arbitrary log-concave random variable in $\RR^n$.
\end{itemize}

We recall that an absolutely continuous probability distribution $\mu$ on $\RR^n$, or a random variable $X \sim \mu$ is log-concave 
if it is supported in an open, convex set $K \susbeteq \RR^n$ (which could be $\RR^n$ itself) and has a density of the form $e^{-H}$ 
in $K$ where $H$ is a convex function. Thus a Gaussian measure is log-concave, as well as the uniform measure on an arbitrary convex body.

\subsection{Main results on $\chi^2$ divergence}

For $\eta_{\chi^2}(\mu,s)$, we prove the following results.
\begin{theorem}
For any probability distribution $\mu$ and $s>0$
\begin{align}
1-\frac{\mmse(\mu,s)}{\EE_\mu|X-\EE X|^2}\leq \eta_{\chi^2}(\mu,s)=S^2(X,X+\sqrt{s}Z)\leq \frac{1}{1+\frac{s}{\CP(\mu)}}.
\label{eq_2110}
\end{align}
If $\mu$ is additionally log-concave, then,
\begin{equation}  \eta_{\chi^2}(\mu,s) \geq e^{-s / C_P(\mu)}. 
	\label{eq_1552}
	\end{equation}
\label{thm:etachi2UB}
\end{theorem}

\begin{theorem}
    For any log-concave probability distribution $\mu$ and function $f:\RR^n\to\RR$, the function $s\mapsto \log\Var(\EE[f(X)|X+\sqrt{s}Z])$ is convex.
In particular, whenever $\mu$ is log-concave, $s\mapsto \log S(X,X+\sqrt{s}Z)$ is convex, and consequently, $s\mapsto \eta_{\chi^2}(\mu,s)$ is also convex.
\label{thm_1355} 
\end{theorem}

Inequality (\ref{eq_1552}) as well as Theorem \ref{thm_1355} are mostly an interpretation of the results of \cite{KP}. Theorem~\ref{thm:etachi2UB} implies that for log-concave $\mu$ the rate at which $s\mapsto \eta_{\chi^2}(\mu,s)$ decreases with $s$ is largely determined by $\CP(\mu)$. In particular, if $s^{\chi^2}_{\mu}(\alpha)=\min\{s~:~\eta_{\chi^2}(\mu,s)\leq \alpha\}$ is the required time  for $\eta_{\chi^2}(\mu,s)$ to drop below $\alpha\in(0,1)$, then
\begin{align}
\CP(\mu)\cdot \log\frac{1}{\alpha}\leq s^{\chi^2}_{\mu}(\alpha)\leq \CP(\mu)\cdot \left(\frac{1}{\alpha}-1\right).
\end{align}
For $\alpha$ close to $1$, the upper bound and lower bound nearly coincide. 
We refer to the quantity $s^{\chi^2}_{\mu}(\alpha)$ for $\alpha = 1/2$ as the ``$\chi^2$ half-blurring time'' of the measure $\mu$.
In the log-concave case, the $\chi^2$ half-blurring time of $\mu$ has the order of magnitude of $C_P(\mu)$, up to an explicit multiplicative universal constant.
A Corollary of Theorem \ref{thm:etachi2UB} is the following:

\begin{corollary}
\begin{align}
\mmse(\mu,s)\geq\frac{\EE_\mu|X-\EE X|^2}{1+\frac{\CP(\mu)}{s}}=\frac{nP}{1+\frac{P}{s}\cdot \nCP(\mu)}
\end{align}
\label{cor:mmseLB}
where
\begin{align}
P=\frac{1}{n}\EE_\mu|X-\EE X|^2,
\label{eq:def_P}
\end{align}
and
\begin{align}
\nCP(\mu)=\frac{\CP(\mu)}{P}.
\label{eq:def_nCP}
\end{align} \label{cor_1405}
\end{corollary}

Let us compare Corollary \ref{cor_1405} with known bounds in Estimation Theory. 
Assume that $\mu$ has smooth density $\rho$ and recall that the Bayesian Cram\'er-Rao lower bound (CRLB)/multivariate van Trees inequality~\cite[eq.10]{gill1995applications} for the additive isotropic Gaussian noise case, gives
\begin{align}
\mmse(\mu,s)\geq \frac{n^2 s}{n+s\m{J}(\mu)}=\frac{nP}{\nFisher(\mu)+\frac{P}{s}},    
\label{eq:CRLB}
\end{align}
where
\begin{align}
 \m{J}(\mu)=\int_{x\in\RR^n}\frac{|\nabla \rho(x)|^2}{\rho(x)}dx. 
\end{align}
and
\begin{align}
\nFisher(\mu)= \frac{\m{J}(\mu)\cdot P}{n}.
\label{eq:def_nFisher}
\end{align}
We see that the lower bound from Corollary~\ref{cor:mmseLB} is tighter than the Bayesian CRLB whenever $\frac{P}{s}(\nCP(\mu)-1)<\nFisher(\mu)-1$. In particular, whenever $\nCP(\mu)$ is finite, the bound from Corollary~\ref{cor:mmseLB} is tighter than the Bayesian CRLB for low enough signal-to-noise ratio (SNR) $\frac{P}{s}$. There are also cases where $\nFisher(\mu)$ is infinite, rendering the Bayesian CRLB useless for all $s>0$, whereas $\nCP(\mu)$ is finite. An example of such probability distribution is the uniform distribution over a convex set.
We remark that there are also examples where $\m{J}(\mu)$ is finite and $\CP(\mu)$ infinite, such as the density $t \mapsto e^{-(t^2 + 1)^{\alpha}}$ on the real line for $0 < \alpha < 1/2$.

\medskip Using the I-MMSE identity of Guo, Shamai and Verd\`u~\cite[Theorem 2]{guo2005mutual}, we can deduce the following.
\begin{corollary}
\begin{align}
I(\mu,s)=I(X;X+\sqrt{s}Z)\geq \frac{n}{2}\cdot\frac{1}{\nCP(\mu)}\log \left(1+\frac{P}{s}\cdot \nCP(\mu) \right),
\end{align}
where $P$ and $\nCP(\mu)$  are as defined in~\eqref{eq:def_P} and~\eqref{eq:def_nCP}, respectively.
\label{cor:mi}
\end{corollary}
It is easy to verify that the function $t\mapsto \frac{1}{t}\log(1+\frac{P}{s}\cdot t)$ is decreasing. Thus, the lower bound above is monotonically decreasing in $\nCP(\mu)$, which in turn is minimized for the isotropic Gaussian distribution, for which its value is $1$. Thus, the lower bound is tight for the isotropic Gaussian distribution, and controls the mutual information loss of the probability distribution $\mu$ from the upper bound $\frac{n}{2}\log(1+\frac{P}{s})$ via the distance from Gaussianity measure $\nCP(\mu)\geq 1$. 

\medskip 
It is insightful to compare the lower bound from Corollary~\ref{cor:mi} to the standard entropy power inequality (EPI)-based~\cite[Theorem 17.7.3]{Cover} lower bound
\begin{align}
I(\mu,s)\geq \frac{n}{2}\log \left(1+\frac{P}{s}\cdot \nEP(\mu)\right).
\label{eq:miEPI}
\end{align}
where
\begin{align}
\nEP(\mu)=\frac{e^{\frac{2}{n}h(\mu)}}{2\pi e\cdot P}
\end{align}
is the normalized entropy power (recall that $\nEP(\mu)\leq 1$ with equality iff $\mu$ is Gaussian isotropic). It is easy to see that whenever $\nCP(\mu)<\infty$, the new lower bound from Corollary~\ref{cor:mi} is tighter than the EPI-based bound~\eqref{eq:miEPI} for $s$ large enough (low SNR).

\medskip Note that using the Bayesian CRLB~\eqref{eq:CRLB} and the I-MMSE identity, one easily obtains the lower bound
\begin{align}
I(\mu,s)\geq \frac{n}{2}\log \left(1+\frac{P}{s}\cdot \frac{1}{\nFisher(\mu)} \right).
\end{align}
However, since for any probability distribution $\mu$ we have that $\nFisher(\mu)\cdot \nEP(\mu)\geq 1$, with equality iff $\mu$ is the isotropic Gaussian distribution~\cite{carlen91},~\cite[Theorem 16]{DemboCoverThomas91}, this bound is subsumed by the EPI-based bound~\eqref{eq:miEPI}.

\subsection{Main results on KL-divergence}

We now move on to present our results concerning the KL divergence. Our first result is a generalization of de Bruijn's identity.
\begin{proposition} Let $\mu$ and $\nu$ be probability measures on $\RR^n$ with $D(\nu \D \mu) < \infty$, and assume $\EE_\mu|X|^4<\infty$. Then for any $s > 0$,
\begin{align}
\frac{d}{ds} D( \nu_s \D \mu_s) = -\frac{1}{2} J( \nu_s \D \mu_s),
\end{align}
where the Fisher information of a probability distribution $\nu$ relative to $\mu$ is, with $f = d\nu / d \mu$,
\begin{align}
 J(\nu \D \mu) = \int_{\RR^n} \frac{|\nabla f|^2}{f} d \mu.
 \end{align}
\label{thm:deBruijnGeneralized}
\end{proposition}

Proposition~\ref{thm:deBruijnGeneralized} easily implies de Bruijn's identity. 

\medskip

Some of the proceeding results rely on the correctness of a certain conjecture, namely Conjecture~\ref{lem_1824}, stated in Section~\ref{sec:KLandGenProofs}. Calling it a conjecture is perhaps too modest: we provide an almost full of proof for Conjecture~\ref{lem_1824}. The only gap between our ``proof'' and a fully rigorous proof is the justification of order exchange between integration and differentiation, as well as justifying several integration by parts (that is, showing that the involved functions vanish at infinity). While it is not trivial to rigorously justify these steps, we are quite certain they are correct under some minimal regularity assumptions. In fact, the justification of such steps is often ignored in the literature. For example, to the best of our knowledge, the first time a fully rigorous proof was given for de Bruijn's identity was by Barron in 1984~\cite{barron1984monotonic}, while prior proofs, such as that of Blachman~\cite{blachman65} did not justify the exchange of differentiation and integration. Whenever a statement relies on the correctness of Conjecture~\ref{lem_1824}, we explicitly emphasize it.

\medskip

We prove the following bounds for the KL contraction coefficient:

\begin{theorem}
For any probability measure $\mu$  and $s>0$,
\begin{align}
1-\frac{\mmse(\mu,s)}{\EE_\mu |X-\EE X |^2}\leq \eta_{\text{KL}}(\mu,s)\leq \frac{1}{1+\frac{s}{\CLS(\mu)}}.
\label{eq_1120}
\end{align}
If additionally $\mu$ is log-concave, then, assuming the validity of Conjecture~\ref{lem_1824}, we also have
\begin{align}
1-\frac{s}{\CLS(\mu)}\leq\eta_{\text{KL}}(\mu,s).
\end{align}
\label{thm:etaKLbounds}
\end{theorem}

The left-hand side inequality in (\ref{eq_1120}) immediately follows from 
(\ref{eq_1135}) and (\ref{eq_2110}).

\begin{theorem}
Assuming validity of Conjecture~\ref{lem_1824}, for any probability distribution $\nu$ and any log-concave probability distribution $\mu$, the mapping $s\mapsto D(\nu_s\|\mu_s)$ is convex. 
In particular, assuming validity of Conjecture~\ref{lem_1824}, whenever $\mu$ is log-concave, $s\mapsto \eta_{\text{KL}}(\mu,s)$ is convex
\label{thm:KLconvex}
\end{theorem}


Theorem~\ref{thm:KLconvex} implies that $s\mapsto h(\nu_s)$ is concave, which is~\cite[Corollary 1]{costa1985new}. Thus, Theorem~\ref{thm:KLconvex} is stronger than Costa'a Corollary about concavity of $s\mapsto h(\nu_s)$. On the other hand, Theorem~\ref{thm:KLconvex} does not imply the concavity of $s\mapsto \exp \left( \frac{2}{n} h(\nu_s) \right)$, which is Costa's main result in~\cite{costa1985new}.

\medskip 
As in our analysis for $\eta_{\chi^2}(s,\mu)$, we see that Theorem~\ref{thm:etaKLbounds} implies that for log-concave $\mu$ the rate at which $s\mapsto \eta_{\text{KL}}(\mu,s)$ decreases with $s$ is largely determined by $\CLS(\mu)$. In particular, if $s^{\text{KL}}_{\mu}(\alpha)=\min\{s~:~\eta_{\text{KL}}(\mu,s)\leq \alpha\}$ is the required time  for $\eta_{\text{KL}}(\mu,s)$ to drop below $\alpha\in(0,1)$, then
\begin{align}
(1-\alpha)\cdot \CLS(\mu)\leq s^{\text{KL}}_{\mu}(\alpha)\leq \left(\frac{1-\alpha}{\alpha}\right)\CLS(\mu).
\end{align}
For $\alpha$ close to $1$, the upper bound and lower bound nearly coincide.
We refer to the quantity $s^{\text{KL}}_{\mu}(\alpha)$ for $\alpha = 1/2$ as the ``$KL$ half-blurring time'' of the measure $\mu$.
In the log-concave case, the $KL$ half-blurring time of $\mu$ has the order of magnitude of $\CLS(\mu)$, up to an explicit multiplicative universal constant.

\subsection{Main results on general divergences}

Fix a smooth, convex function $\varphi: (0, \infty) \to\RR$ with $\varphi(1)=0$ and set $\vphi(0) = \lim_{t \rightarrow 0^+} \varphi(t)$. 
In \cite{chafai}, Chafa\"i introduced 
the concepts of $\vphi$-Sobolev inequalities and $\vphi$-entropy, 
which turn out to fit nicely with our study of the evolution of 
the $\vphi$-divergence $D_{\vphi}(\nu_s \D \mu_s)$. 

\medskip The $\vphi$-Sobolev constant of 
a random vector $X$ in $\RR^n$, denoted by $C_{\vphi}(X)$, is the infimum over all $C \geq 0$ such that for any locally-Lipschitz function $f: \RR^n \rightarrow (0,\infty)$ with $\EE \vphi(X) = 1$,
\begin{equation}  \EE \vphi(f(X)) \leq \frac{C}{2} \cdot \EE \vphi''(f(X)) |\nabla f(X)|^2. \label{eq_1617} \end{equation}
This generalized the definition of the Poincar\'e constant and the log-Sobolev constant. Indeed,
in the case where $\vphi(x) = x \log x$ is the $KL$-divergence, we have $C_{\vphi}(X) = \CLS(X)$, and a simple argument
reveals that when $\vphi(x) = (x-1)^2$ we have $C_{\vphi}(X) = \CP(X)$. When $X \sim \mu$ we set $C_\vphi(\mu) = C_\vphi(X)$. In the general case, our definition (\ref{eq_1617}) is slightly different from the one in Chafa\"i in that we impose the additional requirement that $\EE f(X) = 1$, which implies that $\vphi(\EE f(X)) = 0$.
In Chafa\"i the left-hand side of (\ref{eq_1617}) is replaced by $\EE \vphi(f(X)) - \vphi(\EE f(X))$. In the cases of interest just described, the two definitions coincide by homogeneity.

\medskip 
Given two probability distributions $\mu$ and $\nu$ on $\RR^n$,
with $\nu$ absolutely continuous with a smooth density with respect to $\mu$, we defined the $\vphi$-Fisher information via
$$ J_{\vphi}(\nu \D \mu) = \int_{\RR^n} \vphi''(f) |\nabla f|^2 d \mu, $$
where $f=d\nu/d\mu$.
Thus we have the defining inequality of the $\vphi$-Sobolev constant, 
\begin{equation} D_{\vphi} ( \nu \D \mu ) \leq \frac{C_{\vphi}(\mu)}{2} \cdot J_{\vphi}(\nu \D \mu ). \label{eq_1146} \end{equation}
Some of the properties of the Poincar\'e and log-Sobolev constants generalize to the more general context of the $\vphi$-Sobolev constant. 
The $\vphi$-Sobolev constant is $2$-homogeneous: For any fixed number $\lambda > 0$,
$$ C_{\vphi}(\lambda X) = \lambda^2 C_{\vphi}(X). $$
This follows from the definition (\ref{eq_1617}), by substituting $g(x) = f(\lambda x)$. 
The Poincar\'e constant is the smallest of all $\vphi$-Sobolev constants:
\begin{proposition} Assume that $\vphi''(1) > 0$. Then for any random vector $X$ for which $C_{\vphi}(X) < \infty$,
$$ C_P(X) \leq C_{\vphi}(X). $$
\label{prop_1744}
\end{proposition}

Proposition \ref{prop_1744} must be known, yet for lack of concise reference we provide its proof below.  
The function $1 / \vphi''$ plays an important role in the sequel. When it is a concave function on $(0, \infty)$ -- this is equivalent to condition (H1) from Chafa\"i \cite{chafai} -- it is possible to compute 
the $\vphi$-Sobolev constant of the Gaussian measure. The following proposition follows from the results of Chafa\"i \cite{chafai} and Proposition \ref{prop_1744}. For convenience we provide a proof.

\begin{proposition} Let $\varphi:(0,\infty)\to\RR$ be a smooth, convex function with $\varphi(1)=0$ such that $1 / \vphi''$ is concave. Let $X$ be a standard
Gaussian random vector in $\RR^n$. Then,
$$ C_{\vphi}(X) = 1. $$
\label{prop_1626} 
\end{proposition}

Proposition \ref{thm:deBruijnGeneralized}, the de Bruijn identity, may be generalized even further:

\begin{proposition} Let $\vphi(x)=\vphi_\lambda(x)=x^\lambda -1$, for $\lambda>1$. Let $\mu$ and $\nu$ be probability distributions on $\RR^n$ with $D_{\vphi}(\nu \D \mu) < \infty$, and assume $\EE_\mu|X|^4<\infty$. Then for any $s > 0$,
\begin{align}
\frac{d}{ds} D_{\vphi}( \nu_s \D \mu_s) = -\frac{1}{2} J_{\vphi}( \nu_s \D \mu_s).
\end{align}
\label{thm:deBruijnGeneralized2}
\end{proposition}
This proposition formally implies Proposition~\ref{thm:deBruijnGeneralized} since $D_{\text{KL}}(\nu\|\mu)=\lim_{\lambda\to 1}\frac{1}{\lambda-1}\log(1+D_{\vphi_\lambda}( \nu \D \mu))$. While we have only proved Proposition~\ref{thm:deBruijnGeneralized2} for $\vphi(x)=\vphi_\lambda(x)$, $\lambda>1$, we believe it should be valid for any   smooth, convex function $\varphi:(0,\infty)\to\RR$ with $\varphi(1)=0$. We have only used the assumption that $\vphi(x)=\vphi_\lambda(x)$, $\lambda>1$, for justifying integration under the integral sign, and for justifying integration by parts.
The following is a generalization of Theorem \ref{thm:etaKLbounds}.

\begin{theorem} Let $\vphi(x)=\vphi_\lambda(x)=x^\lambda -1$, for $1<\lambda\leq 2$, such that $1 / \vphi''$ is concave. Then,
for any probability measure $\mu$  and $s>0$,
\begin{align}
1-\frac{\mmse(\mu,s)}{\EE_\mu |X-\EE X |^2}\leq \eta_{\vphi}(\mu,s) \leq \frac{1}{1+\frac{s}{C_{\vphi}(\mu)}}.
\label{eq_1120_}
\end{align}
If additionally $\mu$ is log-concave, then, assuming the validity of Conjecture~\ref{lem_1824}, we also have that
\begin{align}
1-\frac{s}{C_{\vphi}(\mu)}\leq\eta_{\vphi}(\mu,s). \label{eq_1603}
\end{align}
\label{thm_1601}
\end{theorem}

The left-hand side inequality in (\ref{eq_1120_}) immediately follows from 
(\ref{eq_1135}) and (\ref{eq_2110}). 
We proceed with a generalization of Theorem \ref{thm:KLconvex}.

\begin{theorem} Let $\vphi(x)=\vphi_\lambda(x)=x^\lambda -1$, for $1<\lambda\leq 2$, such that $1 / \vphi''$ is concave, such that $1 / \vphi''$ is concave. 
Then, assuming the validity of Conjecture~\ref{lem_1824}, for any probability measure $\nu$ and any log-concave probability measure $\mu$, the mapping $s\mapsto D_{\vphi}(\nu_s\|\mu_s)$ is convex. 
In particular, whenever $\mu$ is log-concave, $s\mapsto \eta_{\vphi}(\mu,s)$ is convex.
\label{thm:phi_convex}
\end{theorem}

For the geometric meaning of the $\vphi$-Sobolev constant, and for the relations to concentration of measure, we refer the reader to \cite[Section 3.4]{chafai}.

\subsection{Related Work}

Contraction coefficients for the Gaussian channel were studied in~\cite{PW15} (see also~\cite{cpw17}), and it was shown that there exist probability distributions $\mu$ with bounded second moment $\EE_\mu|X-\EE X|^2$ for which $\eta_{\chi^2}(\mu,s)=\eta_{\text{KL}}(\mu,s)=1$. Our results show that while a bounded second moment does not imply non-trivial $\eta_{\chi^2}(\mu,s)$ or $\eta_{\text{KL}}(\mu,s)$, bounded Poincar\`e constant $\CP(\mu)$ or log-Sobolev constant $\CLS(\mu)$, respectively, do imply non-trivial contraction coefficient. Moreover, in the log-concave 
case the Poincar\`e constant or log-Sobolev constant are essentially equivalent to the corresponding contraction coefficient.


\medskip For discrete channels, Raginsky~\cite{raginsky16} has estimated the SDPI constants $\eta_{\chi^2}(\mu,K_{Y|X})$ and $\eta_{\text{KL}}(\mu,K_{Y|X})$, as well as for other choices of $\vphi$, of a source $\mu$ and a general discrete channel $K_{Y|X}$ as a function of 
a Poincar\`e constant $\CP(\mu,K_{Y|X})$ or a log-Sobolev constant $\CLS(\mu,K_{Y|X})$, respectively, corresponding to \emph{both} the source $\mu$ and \emph{the channel $K_{Y|X}$}. See~\cite{raginsky16} 
for the precise definitions of those constants and the relations to the contraction coefficients. We stress that the  Poincar\`e constant $\CP(\mu)$ and log-Sobolev constant $\CLS(\mu)$ used in our bounds depend only on the source $\mu$. Furthermore, our bounds only address the Gaussian channels, while the estimates in~\cite{raginsky16} hold for any discrete channel (though this class does not include the Gaussian channel).

\medskip 
The general problem of developing lower bounds, amenable to evaluation, on the MMSE in estimating a random variable $X$ from a dependent random variable $Y$ is a classic topic in information theory, signal processing, and statistics, and so is the special case of estimation in Gaussian noise. Classic references include~\cite{van2004detection,trees2007bayesian,weinstein1988general,ziv1969some}. See also~\cite[Chapter 30]{PW_book}. Our lower bound from Corollary~\ref{cor:mmseLB} varies smoothly with the noise level $s$, similarly to the Bayesian CRLB, and therefore, cannot capture threshold phenomena/phase transitions of the MMSE in $s$. However, to the best of our knowledge, no previous bounds have estimated $\mmse(\mu,s)$ in terms of the Poincar\`e constant $\CP(\mu)$. A recent work by Zieder, Dytso and Cardone~\cite{zieder2022mmse} develops a lower bound on the MMSE for a class of additive channels from $X$ to $Y$, including the Gaussian one. Their bound is expressed in terms of the random variable $\kappa(X)$, where $\kappa(x)=(\CP(W_{Y|X=x}))^{-1}$ is the inverse of the Poincar\`e constant corresponding to the conditional distributions of the channel's output given the input $x$, as well as the variance of the conditional information density between $X$ and $Y$. Our bound from Corollary~\ref{cor:mmseLB} (which holds only for the Gaussian channel), on the other hand, depends \emph{only} on Poincar\`e constant of the \emph{source} $\CP(\mu)$.

\medskip 
The function $s\mapsto h(X+\sqrt{s}Z)$ for isotropic Gaussian $Z$ statistically independent of $X\sim\mu$, has been studied intensively, from the first days of information theory. de Bruijn's identity~\cite{stam1959some,costa1985new} (see also~\cite[Theorem 17.7.2]{Cover}) shows that the derivative of this function with respect to $s$ is $\frac{1}{2}\m{J}(\mu_s)$. 
As explained here, our Proposition~\ref{thm:deBruijnGeneralized} is a stronger and more general identity. de Bruijn's identity was instrumental for proving the entropy power inequality~\cite{stam1959some,blachman65}. Decades later, Costa used de Bruijn's identity for showing that $s\mapsto e^{\frac{2}{n} h(X+\sqrt{s}Z)}$ is concave, which immediately also implies the concavity of $s\mapsto h(X+\sqrt{s}Z)$. Several alternative proofs for Costa's EPI were given in~\cite{dembo1989simple,villani2000short,guo2006proof,rioul2010information,courtade2017strong}.
An important point of view on the evolution of curvature under the heat flow, involving curvature and dimension was developed by Bakry, \'Emery and the Toulouse school, see the book \cite{BGL}.

\medskip 
In Theorem~\ref{thm:KLconvex} we establish the convexity of $s\mapsto D(\nu_s\|\mu_s)$ for log-concave $\mu$, which generalizes the latter, and weaker, statement of Costa. Unfortunately, we were not able to prove convexity of $s\mapsto \log D(\nu_s\|\mu_s)$. Further improvements of the entropy power inequality for the random variable $X+\sqrt{s}Z$ were established by Courtade~\cite{courtade2017strong}. The relation between $I(\mu,s)$ and $\mmse(\mu,s)$, called the I-MMSE relation, which is intimately related to de Bruijn's identity, was discovered by Guo, Shamai and Verd\'u~\cite{guo2005mutual}.
To the best of our knowledge, our Corollary~\ref{cor:mi} is the first lower bound on $I(\mu,s)$ in terms of $\CP(\mu)$. An upper bound in a somewhat similar spirit was derived in~\cite{aras2019family}. Inequalities interpolating between the Poincar\'e and the log-Sobolev inequality appear in Lata\l{}a and Oleszkiewicz \cite{LO}. This is a family of inequalities parameterized by a parameter $1 \leq p \leq 2$.

\medskip In the remainder of this paper, we prove the theorems stated above.

\section{Minimum-squared error and mutual information}

\subsection{Proof of the lower bound in Theorem~\ref{thm:etachi2UB}}

We show that for any pair of random variables $X,Y$ in $\RR^n\times \m{A}$, where $\m{A}$ is some abstract alphabet, it holds that
\begin{align}
S^2(X,Y)\geq 1-\frac{\mmse(X|Y)}{\EE|X-\EE X|^2},
\end{align}
where
\begin{align}
\mmse(X|Y)=\EE|X-\EE[X|Y]|^2.
\end{align}
The idea is to use linear test functions in 
the definition (\ref{eq:HGR_CS}) of the maximal correlation $S(X,Y)$. 
For the case $n=1$, the lower bound was already observed by Re\`nyi~\cite{renyi1959measures}, see also~\cite[eq. 17]{asoodeh2018estimation}. Here we give a proof for the general case.
Let
\begin{align}
i^*=\argmax_{i\in [n]}\frac{\Var(\EE[X_i|Y])}{\Var(X_i)},
\end{align}
and set $\tilde{f}:\RR^n\to \RR$ as 
$$
\tilde{f}(X)=\frac{X_{i^*}-\EE X_{i^*}}{\sqrt{\Var(X_{i^*})}},
$$
such that $\EE\tilde{f}(X)=0$ and $\Var(\tilde{f}(X))=1$.
Since for non-negative numbers $a_1,\ldots,a_n,b_1,\ldots,b_n$ it holds that $\frac{\sum_{i=1}^n a_i}{\sum_{i=1}^n b_i}\leq \max_{i\in[n]}\frac{a_i}{b_i}$, we have that
\begin{align}
\Var(\EE[\tilde{f}(X)|Y])&=\frac{\Var(\EE[X_{i^*}|Y])}{\Var(X_{i^*})}\nonumber\\
&\geq\frac{\sum_{i=1}^n\Var(\EE[X_{i}|Y]}{\sum_{i=1}^n\Var(X_{i})} \\
&=\frac{\sum_{i=1}^n\Var(X_i)-\EE(\Var[X_{i}|Y)]}{\sum_{i=1}^n\Var(X_{i})} \\
&=1-\frac{\mmse(X|Y)}{\EE_\mu |X-\EE[X]|^2}.
\end{align}
The statement follows by invoking~\eqref{eq:HGR_CS}
\begin{align}
S^2(X;Y)=\sup_f \Var(\EE[f(X)|Y])\geq \Var(\EE[\tilde{f}(X)|Y]).
\end{align}

\subsection{Proof of Corollary~\ref{cor:mi}}

Let 
\begin{align}
\tilde{I}(\mu,\rho)&=I(X;\sqrt{\rho}X+Z),\label{eq:tildemi}\\
\widetilde{\mmse}(\mu,\rho)&=\EE|X-\EE[X|\sqrt{\rho}X+Z]|^2.\label{eq:tildemmse}
\end{align}
Clearly, $I(\mu,s)=\tilde{I}(\mu,\frac{1}{s})$ and $\mmse(\mu,s)=\widetilde{\mmse}(\mu,\frac{1}{s})$. By Corollary~\ref{cor:mmseLB}, we therefore have that
\begin{align}
\widetilde{\mmse}(\mu,\rho)= \mmse\left(\mu,\frac{1}{\rho}\right)\geq \frac{nP}{1+\rho \nCP(\mu) P}.
\label{eq:reparMMSElb}
\end{align}
By the I-MMSE identity~\cite[Theorem 2]{guo2005mutual},
\begin{align}
\frac{d}{d\rho} \tilde{I}(\mu,\rho)=\frac{1}{2}\widetilde{\mmse}(\mu,\rho).
\label{eq:vectorIMMSE}
\end{align}
Thus,
\begin{align}
\tilde{I}(\mu,\rho)=\frac{1}{2}\int_{0}^{\rho}\widetilde{\mmse}(\mu,t)dt\geq \frac{1}{2}\int_{0}^{\rho} \frac{nP}{1+t \nCP(\mu) P}dt=\frac{n}{2}\frac{1}{\nCP(\mu)}\log\left(1+\nCP(\mu) P\cdot\rho\right).
\end{align}
where in the last equality we have used the identity $\int_{0}^\rho \frac{a}{1+bt}dt=\frac{a}{b}\log(1+b\rho)$ for $a,b>0$. The claimed result follows by recalling that $I(\mu,s)=\tilde{I}(\mu,\frac{1}{s})$.

\section{$\chi^2$-divergence and proofs for blurring time bounds}
\label{sec2}

In this section we complete the proof of 
Theorem \ref{thm:etachi2UB}. Let $X$ be a random vector attaining values in $\RR^n$.
When proving Theorem \ref{thm:etachi2UB} we may assume that $C_P(X) < \infty$ as otherwise 
the conclusion is vacuous. Recall from above that a log-concave random vector has a finite Poincar\'e constant, and that 
for any additional random vector $Y$,
\begin{equation}  S^2(X,Y) = \sup_f \Var(\EE[f(X) | Y]) 
\label{eq_1609} \end{equation}
where the supreumum runs over all measurable functions $f: \RR^n \rightarrow \RR$ with $\EE f(X) = 0 $ and $\Var(f(X)) = 1$. 
The requirement that $\EE f(X) = 0$ is actually not necessary. Let $Z$ be a standard Gaussian random
vector in $\RR^n$, independent of $X$. Recall that we are interested in the $\chi^2$
contraction coefficient
$$ \eta_{\chi^2}(\mu,s) = S^2(X, X + \sqrt{s} Z). $$
Following Klartag and Putterman \cite{KP}, for a function $f: \RR^n \rightarrow \RR$ with $\EE |f(X)| < \infty$ and for $s > 0$ we write
$$ Q_s f (X + \sqrt{s} Z) = \EE [f(X) | X + \sqrt{s} Z]. $$
Let $\mu$ be the probability distribution on $\RR^n$ that is the distribution law of the random vector $X$,
and let $\mu_s$ be the distribution law of $X + \sqrt{s} Z$. The operator 
$$ Q_s: L^2(\mu) \rightarrow L^2(\mu_s) $$
is of norm at most one, since
$$ \int_{\RR^n} |Q_s f|^2 d \mu_s = \EE \left| Q_s f (X + \sqrt{s} Z) \right|^2 
= \EE \left| \EE [f(X) | X + \sqrt{s} Z] \right|^2 
\leq \EE |f(X)|^2. $$

\begin{lemma} For any $s > 0$, the quantity $S^2(X, X + \sqrt{s} Z)$ is the square of the operator norm of $Q_s: L^2(\mu) \rightarrow L^2(\mu_s)$ restricted 
to the subspace of functions of $\mu$-average zero. In other words, $S^2(X, X + \sqrt{s} Z)$ is the minimal number $M \geq 0$ such that for any $f \in L^2(\mu)$ with $\int f d \mu = 0$,
\begin{equation}  \| Q_s f \|_{L^2(\mu_s)}^2 \leq M \cdot\| f \|_{L^2(\mu)}^2. 
\label{eq_1627} \end{equation} \label{lem_541}
\end{lemma}

\begin{proof} By (\ref{eq_1609}), for any $s > 0$,
$$ S^2(X,X + \sqrt{s} Z) = \sup_f \Var(\EE[f(X) | X + \sqrt{s} Z]) = \sup_f  \Var(Q_s f(X + \sqrt{s} Z)), $$
where the supremum runs over all $f$ with $\EE f(X) = 0$ and $\Var f(X) = 1$. Therefore,
 $$ S^2(X,X + \sqrt{s} Z) = \sup_f \int_{\RR^n} |Q_s f|^2 d \mu_s = \sup_f \| Q_s f \|_{L^2(\mu_s)}^2, $$
where the supremum runs over all $f$ with $\int f d \mu = 0$ and $\int f^2 d \mu = 1$. 
\end{proof}

The operator $Q_s$ may be expressed as an integral operator. Write 
\begin{equation}  \gamma_s(x) = (2 \pi s)^{-n/2} 
\exp(-|x|^2 / (2s)) \label{eq_1716} \end{equation} for the density in $\RR^n$ of a Gaussian random vector of mean zero and covariance $s \cdot \id$.
By considering the joint distribution of $(X, X + \sqrt{s} Z)$
and writing conditional expectation as an integral, we see that
\begin{equation}  Q_s f(y) = \frac{\int_{\RR^n} \gamma_s(x - y) f(x) d \mu(x)}{(\rho * \gamma_s)(y)}. \label{eq_1125} \end{equation}
The proof of Theorem \ref{thm:etachi2UB} requires two differentiations of  the expression on the left-hand side 
of (\ref{eq_1627}) with respect to $s$. It will be convenient to consider a subclass 
of well-behaved functions of $L^2(\mu)$ in order to justify the differentiations under the integral sign.
Write $\rho$ for the density of $\mu$.

\medskip 
As in Klartag and Putterman \cite{KP}, 
we say that a function $f: \RR^n \rightarrow \RR$ has subexponential decay relative to $\rho$
if there exist $C, a > 0$ such that
\begin{equation} |f(x)| \leq \frac{C}{\sqrt{\rho(x)}} e^{-a|x|} \qquad \qquad \qquad (x \in \RR^n).
\label{eq_1030} \end{equation}
If $\rho$ decays exponentially at infinity -- for instance, if $\rho$ is log-concave  --
then all polynomials have subexponential decay relative to $\rho$.
We say that a function $f: \RR^n \rightarrow \RR$ is $\mu$-{\it tempered} if it is smooth and if all of its partial derivatives
of all orders have subexponential decay relative to $\rho$.
Write $\cF_{\mu}$ for the collection of all $\mu$-tempered functions on $\RR^n$. Since $\cF_{\mu}$ contains 
all compactly-supported smooth functions, it is a dense subspace of $L^2(\mu)$.

\begin{lemma} If $f$ is $\mu$-tempered, then $Q_s f$ is $\mu_s$-tempered for any $s > 0$.
\label{lem_526}
\end{lemma}

Lemma \ref{lem_526} is proven in \cite[Lemma 2.2]{KP} under the additional assumption that $\mu$ is log-concave. The log-concavity assumption is only used in the proof of Lemma 2.2 in \cite{KP} in order to show that for any $p > 0$,
\begin{equation}  \int_{\RR^n} |x|^p d \mu(x) < \infty. 
\label{eq_532} \end{equation}
However, a measure $\mu$ with $C_P(\mu) < \infty$ clearly satisfies (\ref{eq_532}). In fact, since $f(x) = |x|$ is a $1$-Lipschitz
function, we even know that $\int_{\RR^n} e^{\alpha f} d \mu < \infty$ for some positive $\alpha > 0$, this goes back to Gromov and Milman \cite{GM}.
Hence Lemma \ref{lem_526} applies for any probability measure $\mu$ with a finite Poincar\'e constant.

\medskip 
We write $\nabla^2 f(x) \in \RR^{n \times n}$
for the Hessian matrix of $f$ at the point $x$.
We abbreviate $X_s = X + \sqrt{s} Z$. 
The following lemma
is proven in \cite{KP} as well:

\begin{lemma} For any $s > 0$ and $f \in \cF_{\mu}$, setting $f_s = Q_s f$,
\begin{equation}  \frac{\partial}{\partial s} \Var ( f_s(X_s) ) = - \EE |\nabla f_s(X_s)|^2,
 \label{eq_1729} \end{equation}
and
\begin{equation}
\frac{\partial}{\partial s} \EE |\nabla f_s(X_s)|^2 = - \EE \|\nabla^2 f_s(X_s) \|_{HS}^2 - 2 \EE (\nabla^2 \psi_s)(X_s) \nabla f_s (X_s) \cdot \nabla f_s(X_s),
\label{eq_916}
\end{equation}
where $\rho_s = e^{-\psi_s}$ is the density of $X_s$ and $\| A \|_{HS}$ is the Hilbert-Schmidt norm (also known as the Frobenius norm) of the matrix $A \in \RR^{n \times n}$.
\end{lemma}

Recall that the density of the random vector $X = X_0$ is the function $\rho = e^{-\psi}$.
The Laplace operator associated with $\mu$ is defined for $u \in \cF_{\mu}$ via
\begin{equation}  L u = L_{\mu} u = \Delta u - \nabla \psi \cdot \nabla u. \label{eq_1647} \end{equation}
For any $u,v \in \cF_{\mu}$ we have the integration by parts formula
\begin{equation} \int_{\RR^n} (L u) v d \mu = -\int_{\RR^n} \langle \nabla u, \nabla v \rangle d \mu
\label{eq_1047_} \end{equation}
and the (integrated) Bochner formula
\begin{equation}
\int_{\RR^n} (L u)^2 d \mu = \int_{\RR^n} \| \nabla^2 u \|_{HS}^2 d \mu + \int_{\RR^n} \langle (\nabla^2 \psi) \nabla u, \nabla u \rangle d \mu,
\label{eq_1531}
\end{equation}
where $\| \nabla^2 u \|_{HS}$ is the Hilbert-Schmidt norm of the Hessian matrix $\nabla^2 u$.
Formulae (\ref{eq_1047_}) and (\ref{eq_1531}) are proven by intergation by parts,
see e.g., Ledoux \cite[Section 2.3]{ledoux1}. The $\mu$-temperedness of $u,v$ are used in order to discard the boundary terms. The operator $-L$ is symmetric and positive semi-definite in $\cF_{\mu} \subseteq L^2(\mu)$. 
Recall the well-known subadditivity property of the Poincar\'e constant (e.g. \cite{C} or Theorem \ref{thm_1315}),
$$ \CP(X_s) = \CP(X + \sqrt{s} Z) \leq \CP(X) + \CP(\sqrt{s} Z) = \CP(X) + s. $$

\begin{proof}[Proof of Theorem \ref{thm:etachi2UB}] We begin with the proof of the right-hand side inequality in (\ref{eq_2110}), 
as the left-hand side inequality was already proven above. 
If $C_P(X) = +\infty$ then this inequality is vacuously true, hence we may assume that $C_P(X) < \infty$. 
Recall that $X_s = X + \sqrt{s} Z$. By Lemma \ref{lem_541}, we need to show that for any function $f$ with $\EE |f(X)|^2 < \infty$,
\begin{equation} \Var \left( \EE \left[ f(X) | X_s \right] \right)  = \Var(Q_s f(X_s)) \leq \frac{1}{1 + s / C_P(\mu)}  \cdot \Var f(X). 
\label{eq_543} \end{equation}
Recall that $\mu$ is the distribution law of $X$.
Since $Q_s: L^2(\mu) \rightarrow L^2(\mu_s)$ is a bounded operator and $\cF_{\mu}$ is dense in $L^2(\mu)$,
it suffices to prove (\ref{eq_543}) under the additional assumption that $f \in \cF_{\mu}$. Abbreviate $f_s = Q_s(f)$.
By using (\ref{eq_1729}) and the Poincar\'e inequality, we get 
$$
\frac{\partial}{\partial s} \Var(f_s(X_s)) \leq -\frac{1}{\CP(X_s)} \Var(f_s(X_s)) \leq -\frac{1}{s +\CP(X)} \Var(f_s(X_s)),
$$
where we used  subadditivity in the last passage. Therefore
$$ \frac{\partial}{\partial s} \log \Var(f_s(X_s)) \leq -\frac{1}{s +\CP(X)}.
$$
By integrating from $0$ to $s$ we conclude that 
$$\log \frac{\Var(f_s(X_s))}{\Var f(X)}   \leq
-\int_0^s \frac{dx}{x + \CP(X)}  = \log \frac{1}{1 + s/\CP(X)}, $$
proving (\ref{eq_543}). This proves the right-hand side inequality in (\ref{eq_2110}).
Let us now assume that $X$ is log-concave and prove (\ref{eq_1552}). This is in fact 
proven in \cite{KP}, but let us briefly repeat the argument here for convenience. Let $0 < \eps < \CP(\mu)$. By the definition of $\CP(\mu)$, 
there exists a non-constant function $f$ such that 
\begin{equation}  \Var f(X) \geq (\CP(\mu) - \eps) \cdot \EE |\nabla f(X)|^2. \label{eq_610} \end{equation}
Thanks to the appendix of \cite{BK}, we may assume that $f$ is smooth and compactly-supported in $\RR^n$, and in particular $f \in \cF_{\mu}$. 
According to (\ref{eq_1729}) and (\ref{eq_916}) we may differentiate the Rayleigh quotient
$$ \frac{\partial}{\partial s} \frac{\EE |\nabla f_s(X_s)|^2}{\Var (f_s(X_s))} = - \EE |\nabla^2 \tilde{f}_s(X_s)|^2 - 2 \EE (\nabla^2 \psi_s)(X_s) \nabla \tilde{f}_s (X_s) \cdot \nabla \tilde{f}_s(X_s) - \EE |\nabla \tilde{f_s}(X_s)|^4 $$
where we normalize  $\tilde{f}_s = f_s / \sqrt{\Var(f_s(X_s)}$. As explained in the proof of \cite[Theorem 2.4]{KP}, it follows from the spectral theorem that 
$$ \langle L_{\mu_s}^2 \tilde{f}_s, \tilde{f}_s \rangle_{L^2(\mu_s)} \geq \langle L \tilde{f}_s, \tilde{f}_s \rangle_{L^2(\mu_s)}^2
= \left( \int_{\RR^n} |\nabla \tilde{f}_s|^2 d\mu_s \right)^2. $$
Therefore, from the Bochner formula (\ref{eq_1531}),
\begin{equation} 
 \frac{\partial}{\partial s} \frac{\EE |\nabla f_s(X_s)|^2}{\Var (f_s(X_s))} \leq -  \EE (\nabla^2 \psi_s)(X_s) \nabla \tilde{f}_s (X_s) \cdot \nabla \tilde{f}_s(X_s) \leq 0. \label{eq_1818}
 \end{equation}
where the last inequality is the only place where we use log-concavity; indeed, 
since $X$ is log-concave,  so is the random vector $X + \sqrt{s} Z$ by the Pr\'ekopa-Leindler 
inequality (see e.g. the first pages of Pisier \cite{Pis}, or Davidovi\v{c}, Korenbljum and Hacet \cite{DKH}). Thus $e^{-\psi_s}$ is a log-concave density, which amounts to the fact that the symmetric matrix $\nabla^2 \psi_s$ is positive semi-definite. Hence the term involving the Hessian $\nabla^2 \psi_s$ in (\ref{eq_1818}) is non-positive. 
Consequently,
$\EE |\nabla f_s(X_s)|^2 / \Var(f_s(X_s))$ is non-increasing in $s$ and
$$ \log \frac{\Var(f_s(X_s))}{\Var f(X)} = -\int_0^s \frac{\EE |\nabla f_t(X_t)|^2}{\Var (f_t(X_t))} dt \geq -s \cdot \frac{\EE |\nabla f(X)|^2}{\Var (f(X))} \geq -\frac{s}{\CP(\mu) - \eps}. $$
Hence for any $\eps > 0$ we found $f$ such that 
$$ \Var \left(\EE \left[ f(X) | X_s \right] \right) = \Var(f_s(X_s)) \geq \exp \left( -\frac{s}{\CP(\mu) - \eps} \right) \Var f(X). $$
Since $\eps > 0$ is arbitrary, and in view of Lemma \ref{lem_541} this proves that $S^2(X, X + \sqrt{s} Z) \geq \exp(-s / \CP(\mu))$.
Inequality (\ref{eq_1552}) is thus proven.
\end{proof}

\begin{remark} Thomas Courtade and Joseph Lehec communicated to us an alternative  proof of the right-hand side inequality in (\ref{eq_2110}). Their proof does not require differentiation with respect to the parameter $s$.
\end{remark}

\section{Identities for KL-divergence and general divergences}
\label{sec:KLandGenProofs}

In this section we prove Propositions \ref{thm:deBruijnGeneralized} and \ref{thm:deBruijnGeneralized2}, the generalized de Bruijn identity.
We keep the notation 
from the previous section. Let $f: \RR^n \rightarrow \RR$ be a $\mu$-integrable function.

\medskip As in the previous section, some care is needed when differentiating under the integral sign and when integrating by parts and neglecting the boundary terms. As opposed to the previous section analyzing the case of the $\chi^2$-divergence, here we will be rather brief in the justification. As in \cite{KP} and in Section \ref{sec2}, the basic idea is to introduce a suitable class $\cF_{\mu, \vphi}$ of smooth functions, which allow for integrations by parts without boundary terms, and is preserved by the $Q_s$-dynamics. Then one approximates the density of an arbitrary measure $\nu$ with $J_{\vphi}(\nu \D \mu) < \infty$ with a function from this class.

\medskip From Klartag and Putterman \cite{KP}, we know that for any $f \in \cF_{\mu}$ and $s > 0$,
\begin{equation}  \frac{d}{ds} Q_s f = \Box_s Q_s f \label{eq_1821} \end{equation}
where
$$ \Box_s = L_s - \frac{\Delta}{2} $$
and $L_s = L_{\mu_s}$ is the Laplace operator associated with the measure $\mu_s$.
We write $\rho_s$ for the density of $\mu_s$.

\begin{lemma} Let $\vphi(x)=\vphi_\lambda(x)=x^\lambda -1$, for $\lambda>1$. Let $\mu$ and $\nu$ be probability distributions on $\RR^n$ with $D_{\vphi}(\nu \D \mu) < \infty$, and assume $\EE_\mu|X|^4<\infty$.  Denote $f=d\nu/d\mu$  such that $f_s = Q_s f=d\mu_s/d\nu_s$,~for $s > 0$. we have 
\begin{align}
\frac{d}{ds} \EE \vphi( f_s(X_s) ) = -\frac{1}{2} \cdot \EE \vphi''(f_s(X_s)) |\nabla f_s(X_s)|^2 = -\frac{1}{2} \int_{\RR^n} \vphi''(f_s) |\nabla f_s|^2 d\mu_s.\label{eq:generalderivative}    
\end{align}
 In particular, if $D_{\text{KL}}(\nu \D \mu) < \infty$
\begin{equation} \frac{d}{ds} \EE f_s(X_s) \log f_s(X_s) = -\frac{1}{2} \cdot \EE \frac{|\nabla f(X_s)|^2}{f_s(X_s)} \label{eq_1713} \end{equation}

\label{lem_620}
\end{lemma}

\begin{proof} 
We prove~\eqref{eq:generalderivative} below for any smooth, convex function $\vphi:(0,\infty)\to\RR$ with $\vphi(1)=0$, without justifying differentiation under the integral sign, and without justifying integration by parts. The justifications for these steps for $\vphi(x)=\vphi_\lambda(x)=x^\lambda -1$, for $\lambda>1$, under the assumptions that $D_{\vphi}(\nu \D \mu) < \infty$, and  $\EE_\mu|X|^4<\infty$, are given in the appendix. We also prove in the appendix that~\eqref{eq:generalderivative} implies~\eqref{eq_1713}. 

\medskip 
Recall that $\partial \rho_s / \partial s = \Delta \rho_s / 2$ by the heat equation, as $\mu_s = \mu * \gamma_s$ and $\rho_s$ is the density of $\mu_s$.
This and (\ref{eq_1821}) are the main places where we use the heat equation in our arguments. Let us compute:
\begin{align*}
\frac{d}{ds} \int_{\RR^n} \vphi(f_s) \rho_s = \int_{\RR^n} \vphi'(f_s) \Box_s f_s \cdot \rho_s + \int_{\RR^n} \vphi(f_s) \frac{\Delta \rho_s}{2}.
\end{align*}
 Integrating by parts we continue with
\begin{align*}
\\  & = \int_{\RR^n} \left[ \vphi'(f_s) \left(L_s f_s - \frac{\Delta f_s}{2} \right)    + \frac{\Delta( \vphi(f_s) )}{2} \right] \rho_s  
\\  & = \int_{\RR^n} \left[ \vphi'(f_s) L_s f_s - \vphi'(f_s) \frac{\Delta f_s}{2}  + \frac{\vphi'(f_s) \Delta f_s + \vphi''(f_s) |\nabla f_s|^2}{2} ) \right] \rho_s  
\\  & = \int_{\RR^n} \left[ -\vphi''(f_s) |\nabla f_s|^2  + \frac{\vphi''(f_s) |\nabla f_s|^2)}{2} \right] \rho_s,
\end{align*}
where we used the integration by parts 
(\ref{eq_1047_}) in the last passage.
\end{proof}

\begin{proof}[Proof of Proposition \ref{thm:deBruijnGeneralized} and Proposition \ref{thm:deBruijnGeneralized2}] 
Follow directly from Lemma \ref{lem_620}, and the definition of $D_{\vphi}$ and $J_{\vphi}$.
\end{proof}

We move on to computing the second derivative of $\EE \vphi(f(X_s))$. 

\begin{conjecture} Denoting
$$ \kappa = -\frac{1}{2 \vphi''} \left( \frac{1}{\vphi''} \right)'', $$
and
$$ g_s = \vphi'(f_s) $$
we have, under mild regularity assumptions, the ``Bochner-type'' formula
\begin{align} 
\nonumber \frac{d}{ds}  \int_{\RR^n} \vphi''(f_s) |\nabla f_s|^2 d \mu_s  = -\int_{\RR^n}  \left[ 2 (\nabla^2 \psi_s) \nabla g_s \cdot \nabla g_s + \left|\nabla^2 g_s \right|^2 + \kappa(f_s) |\nabla g_s|^4  \right] \frac{d \mu_s}{\vphi''(f_s)},
  \end{align}
where  we recall that $\rho_s = e^{-\psi_s}$.
\label{lem_1824}
\end{conjecture}

\begin{proof}[Sketch of proof]
We assume below that differentiation under the integral sign is valid, and that all functions involved in integration in parts vanish. Justifying these assumptions is the missing step in proving this conjecture. We now prove the ``Bochner-typ'' formula above, under the assumption that those steps are indeed valid (i.e., ignoring all regularity issues).

\medskip Our approach follows the $M$-function of Ivanishvili and Volberg \cite{IV} and the dynamic approach to the $\Gamma$-calculus is presented in \cite{KP}. 
We consider a function $M(x,y)$ of two real variables of the form
\begin{equation}  M(x,y) = y \vphi''(x). \label{eq_1723} \end{equation}
Thus $M(x,y) = y/x$ in the case where $\vphi(x) = x \log x$ and $M(x,y) = 2 y$ if $\vphi(x) = (x-1)^2$. In view of Lemma \ref{lem_620}, we would need to compute
$$ \frac{d}{ds} \EE M(f_s(X_s), |\nabla f_s(X_s)|^2). $$
By the chain rule, for any functions $f,g$,
\begin{equation}  \Delta M(f,g) = M_x \Delta f + M_y \Delta g + M_{xx} |\nabla f|^2 + 2 M_{xy} \nabla f \cdot \nabla g + M_{yy} |\nabla g|^2, \label{eq_1846} \end{equation}
where we abbreviate $M_x = \partial M / \partial x, M_y = \partial M / \partial y$ etc. Let us also introduce the dynamic $\Gamma$-calculus notation from \cite{KP}. We set
$$ \Gamma_1(u,v) = \nabla u \cdot \nabla v $$
and
$$  \Gamma_2(u,v) := \Box_s \Gamma_1(u, v) - \Gamma_1(\Box_s u, v ) -
\Gamma_1( u, \Box_s v) $$ 
which satisfies 
\begin{equation} \Gamma_2(u,u) = \| \nabla^2 u \|_{HS}^2 - 2 \nabla^2 (\log \rho_s) \nabla u \cdot \nabla u. \label{eq_1706} \end{equation}
We abbreviate $\Gamma_1(u) = \Gamma_1(u,u)$ and $\Gamma_2(u) = \Gamma_2(u,u)$.

\medskip We first prove that for a general function $M$ of two variables, 
	\begin{align}\label{eq_1052} \frac{d}{ds} & \EE M( f_s(X_s), |\nabla f_s(X_s)|^2)  
		\\& = -\int_{\RR^n} \left[ 	 M_y \Gamma_2(f_s) + \frac{M_{xx}}{2} \Gamma_1(f_s) + M_{xy} \Gamma_1( f_s, \Gamma_1(f_s)) + \frac{M_{yy}}{2} \Gamma_1(\Gamma_1(f_s))\right] d \mu_s. \nonumber 
	\end{align}
To this end, abbreviate $M = M(f_s, |\nabla f_s|^2), M_x = M_x(f_s, |\nabla f_s|^2)$
	etc. Then,
\begin{align} 
\label{eq_1124}
	 \frac{d}{ds} \int_{\RR^n} M(f_s, |\nabla f_s|^2) \rho_s &= \int_{\RR^n} \left[ \frac{\Delta M}{2}  
	 + M_x \Box_s f_s + 2 M_y \Gamma_1( \Box_s f_s, f_s) \right]  \rho_s \\ & = 
	 \int_{\RR^n} \left[ \frac{\Delta M}{2}  
	 + M_x \Box_s f_s + M_y \left( \Box_s \Gamma_1( f_s, f_s) - \Gamma_2(f_s, f_s)\right) \right]  \rho_s
\nonumber
\\ & = 
	 \int_{\RR^n} \left[ \frac{\Delta M}{2}  
	 + M_x \left(L_s - \frac{\Delta}{2} \right)  f_s + M_y \left( \left(L_s - \frac{\Delta}{2} \right) \Gamma_1( f_s) - \Gamma_2(f_s)\right) \right]  \rho_s.
\nonumber
\end{align}
By (\ref{eq_1846}) with $f = f_s, g = \Gamma_1(f_s)$, 
$$ \Delta M = M_x \Delta f_s + M_y \Delta \Gamma_1(f_s) + M_{xx} \Gamma_1(f_s) + 2 M_{xy} \Gamma_1( f_s, \Gamma_1(f_s)) + M_{yy} \Gamma_1(\Gamma_1(f_s)). $$
Therefore,
$$ \frac{1}{2} \left( \Delta M - M_x \Delta f_s - M_y \Delta \Gamma_1(f_s)  \right) =  \frac{M_{xx}}{2} \Gamma_1(f_s) + M_{xy} \Gamma_1( f_s, \Gamma_1(f_s)) + \frac{M_{yy}}{2} \Gamma_1(\Gamma_1(f_s)).
$$
Consequently, the integral in (\ref{eq_1124}) equals
$$ \int_{\RR^n} \left[ 	 M_x L_s f_s + M_y L_s \Gamma_1( f_s) - M_y \Gamma_2(f_s) + \frac{M_{xx}}{2} \Gamma_1(f_s) + M_{xy} \Gamma_1( f_s, \Gamma_1(f_s)) + \frac{M_{yy}}{2} \Gamma_1(\Gamma_1(f_s))\right]  \rho_s $$
Integrate by parts the two terms involving $L_s$ using (\ref{eq_1047_}) to obtain
\begin{align*} \int_{\RR^n} & \left[ 	 M_x L_s f_s + M_y L_s \Gamma_1( f_s) \right] \rho_s = 
\int_{\RR^n} \left[ 	 -\nabla M_x \cdot \nabla f_s - \nabla M_y \cdot \nabla \Gamma_1( f_s) \right] \rho_s
\\ & = \int_{\RR^n} \left[ -(M_{xx} \nabla f_s + M_{xy} \nabla \Gamma_1(f_s)) \cdot \nabla f_s - (M_{xy} \nabla f_s + M_{yy} \nabla \Gamma_1(f_s) )\cdot \nabla \Gamma_1( f_s) \right] \rho_s
\\ & = \int_{\RR^n} \left[ -M_{xx} \Gamma_1(f_s) - 2 M_{xy} \Gamma_1( \Gamma_1(f_s), f_s) - M_{yy} \Gamma_1( \Gamma_1(f_s) ) \right] \rho_s.
\end{align*}
Hence, the integral in (\ref{eq_1124}) equals
$$ -\int_{\RR^n} \left[ 	 M_y \Gamma_2(f_s) + \frac{M_{xx}}{2} \Gamma_1(f_s) + M_{xy} \Gamma_1( f_s, \Gamma_1(f_s)) + \frac{M_{yy}}{2} \Gamma_1(\Gamma_1(f_s))\right]  \rho_s $$
and (\ref{eq_1052}) is proven. In the specific case where $M$ is the function given by (\ref{eq_1723}) we conclude from  (\ref{eq_1052})  that 
\begin{align} 
\label{eq_1858} \frac{d}{ds} & \int_{\RR^n} \vphi''(f_s) |\nabla f_s|^2 d \mu_s 
 \\ & = -\int_{\RR^n} \left[ 	 \vphi''(f_s) \Gamma_2(f_s) + \frac{\vphi^{(4)}(f_s)}{2} \Gamma_1(f_s)^2 + \vphi^{(3)}(f_s) \Gamma_1( f_s, \Gamma_1(f_s))  \right] d \mu_s  \nonumber
 \\ & = -\int_{\RR^n} \left[ 2 \vphi'' (\nabla^2 \psi_s) \nabla f_s \cdot \nabla f_s +	 \vphi'' |\nabla^2 f_s|^2 + \frac{\vphi^{(4)}}{2} |\nabla f_s|^4 + 2 \vphi^{(3)} (\nabla^2 f_s) \nabla f_s \cdot \nabla f_s  \right] d \mu_s \nonumber
 \\ & = -\int_{\RR^n} \left[ 2 \vphi'' (\nabla^2 \psi_s) \nabla f_s \cdot \nabla f_s +	 \vphi'' \left|\nabla^2 f_s + \frac{\vphi^{(3)}}{\vphi''} \nabla f_s \otimes \nabla f_s \right|^2 + \left \{ \frac{\vphi^{(4)}}{2} - \frac{\left( \vphi^{(3)} \right)^2}{\vphi''} \right \} |\nabla f_s|^4  \right] d \mu_s \nonumber
  \end{align}
Note that 
$$
 \frac{\vphi^{(4)}}{2} - \frac{\left( \vphi^{(3)} \right)^2}{\vphi''} = -\frac{(\vphi'')^2}{2} \left( \frac{1}{\vphi''} \right)'' = (\vphi'')^3 \kappa.
$$
Consequently,
\begin{align} 
	\nonumber \frac{d}{ds} & \int_{\RR^n} \vphi''(f_s) |\nabla f_s|^2 d \mu_s 
	\\ & = -\int_{\RR^n} \left[ 2 \vphi'' (\nabla^2 \psi_s) \nabla f_s \cdot \nabla f_s +	 \vphi'' \left|\nabla^2 f_s + \frac{\vphi^{(3)}}{\vphi''} \nabla f_s \otimes \nabla f_s \right|^2 + \kappa (\vphi')^3 |\nabla f_s|^4  \right] d \mu_s. \nonumber
\end{align}
Since $g_s = \vphi'(f_s)$ we have $\nabla g_s = \vphi''(f_s) \nabla f_s$ and $\nabla^2 g_s = \vphi''(f_s) \nabla^2 f_s + \vphi^{(3)}(f_s) \nabla f_s \otimes \nabla f_s$, and the lemma is proven.
\end{proof}

Specializing Conjecture \ref{lem_1824} to the case where $\vphi(x) = x \log x$,  we obtain the following:

\begin{proposition}
Assuming the validity of Conjecture~\ref{lem_1824}, With $g_s = \log Q_s f$ we have
$$  \frac{d}{ds} \EE \frac{|\nabla f_s(X_s)|^2}{f_s(X_s)} = -\int_{\RR^n} \left[ |\nabla^2 g_s|^2 + 2 \langle (\nabla^2 \psi_s) \nabla g_s, \nabla g_s \rangle   \right] f_s d \mu_s. $$
\end{proposition}

\begin{proof} This follows from Conjecture \ref{lem_1824} with $\vphi(x) = x \log x$, since in this case $1 / \vphi''(x) = x$ and $\kappa \equiv 0$.
\end{proof}

\section{Proofs of inequalities related to $\vphi$-Sobolev constants}
\label{sec_ineq}

We keep the notation and assumptions of the previous section. 
We need to use the subadditivity property of the $\vphi$-Sobolev constants, proven by Chafa\"i \cite[Corollary 3.1]{chafai} in greater generality. For the convenience of the reader we include here a statement and a proof
of the following proposition, as well as of Proposition \ref{prop_1626}, which also goes back to Chafa\"i \cite{chafai}.

\begin{theorem}[``Subadditivity of the $\vphi$-Sobolev constants'', Chafa\"i \cite{chafai}] Let $X$ and $Y$ be independent random vectors in $\RR^n$. 
Let $\varphi:(0,\infty)\to\RR$ be a smooth, convex function with $\varphi(1)=0$ such that $1 / \vphi''$ is concave. Then,
$$ C_{\vphi}(X + Y) \leq C_{\vphi}(X) + C_{\vphi}(Y). $$ 
\label{thm_1315}
\end{theorem}

\begin{proof} for any locally-Lipschitz function $f: \RR^n \rightarrow [0, \infty)$
with $\EE |f(X+Y)| < \infty$, denoting $g(x) = \EE f(x + Y)$ we have
\begin{align*} \EE & \vphi(f(X+Y))  - \vphi(\EE(f(X+Y))) \\ & = \EE_X \left[ \EE_Y \vphi(f(X+Y)) - \vphi(\EE_Y \vphi(f(X+Y))) \right] + \EE_X \vphi(g(X)) - \vphi(\EE_X g(X)) 
\\ & \leq \frac{C_{\vphi}(Y)}{2} \EE_X \EE_Y \vphi''(f(X+Y)) |\nabla f(X+Y)|^2 + \frac{C_{\vphi}(X)}{2} \EE_X \vphi''(g(X)) |\nabla g(X)|^2. 
\end{align*}
To conclude the proof it remains to show that for any fixed $x \in \RR^n$,
$$
\vphi''(g(x)) |\nabla g(x)|^2  \leq \EE_Y \vphi''(f(x+Y)) |\nabla f(x+Y)|^2. $$
Indeed, this would imply that 
$$ \EE \vphi(f(X+Y))  - \vphi(\EE(\vphi(X+Y))) \leq \frac{C_{\vphi}(Y) + C_{\vphi}(X)}{2} \EE \vphi''(f(X+Y)) |\nabla f(X+Y)|^2,  $$
and hence $C_{\vphi}(X+Y) \leq C_{\vphi}(Y) + C_{\vphi}(X)$. Since $\nabla g(x) = \EE \nabla f(x + Y)$, all that remains is to show that 
$$
\vphi''(\EE f(x + Y)) |\EE \nabla f(x + Y)|^2 \leq \EE_Y \vphi''(f(x+Y)) |\nabla f(x+Y)|^2. $$
By Jensen's inequality, this would follow once we show that the function 
\begin{equation}  (x,y) \mapsto \vphi''(x) |y|^2 \qquad \qquad (x >0, y \in \RR^n) \label{eq_1547} \end{equation}
is jointly convex in $x$ and $y$. Since $\vphi''(x) |y|^2 = \sum_j \vphi''(x) y_j^2$, it suffices to show that 
the function 
\begin{equation}  (x,y) \mapsto \vphi''(x) y^2 \qquad \qquad (x >0, y \in \RR) \label{eq_1547_} \end{equation}
is jointly convex in $(x,y)$. The Hessian of this function is the matrix
$$ \left( \begin{matrix} \vphi^{(4)}(x) y^2 & 2 \vphi^{(3)}(x) y \\ 2 \vphi^{(3)}(x) y  & 2 \vphi''(x) \end{matrix} \right) $$
Since this symmetric $2 \times 2$ matrix has a non-negative entry on the diagonal, namely $2 \vphi''(x)$, by the Sylvester criterion it is 
positive  semi-definite if and only if its determinant is non-negative. The determinant equals
$$ 2 y^2 \left[ \vphi'' \vphi^{(4)} - 2 (\vphi^{(3)})^2 \right] = -2 y^2 (\vphi'')^3 \left( \frac{1}{\vphi''} \right)'', $$
where we used the formula $(1 / \vphi'')'' = -(\vphi'' \vphi^{(4)} - 2 (\vphi^{(3)})^2 ) / (\vphi'')^3$.
Since $\vphi'' \geq 0$ and $1 / \vphi''$ is concave, the determinant is non-negative, and the function in (\ref{eq_1547_}) is convex. This completes the proof.
\end{proof}

\begin{proof}[Proof of Proposition \ref{prop_1744}] For any locally-Lipschitz function $f$, 
	$$ \EE \vphi(f(X)) - \vphi( \EE f(X) ) \leq \frac{C_{\vphi}(X)}{2} \cdot \EE \vphi''(f(X)) |\nabla f(X)|^2. $$
Let $g: \RR^n \rightarrow \RR$ be a  Lipschitz, bounded function. Then there exists $\eps_0 > 0$ such that for any $0 < \eps < \eps_0$, the function $g_0 = 1 + \eps g$ attains values at $I$. In fact, since $\vphi$ is smooth, there exists $M > 0$ and $\eps_1 < \eps_0$ such that for $\eps < \eps_1$,
$$ \left| \vphi(g_0(x)) - \eps \vphi'(1) g(x) -\frac{\eps^2}{2} \vphi''(1) g^2(x) \right| \leq \eps^3 M $$
and
$$ \vphi''(g_0(x)) \geq (1 - \eps M) \cdot \vphi''(1). $$
Consequently, for any $0 < \eps < \eps_1$,
$$ \frac{\eps^2}{2} \vphi''(1) Var(g(X)) - 2 M \eps^3 \leq \EE \vphi(f(X)) - \vphi( \EE f(X) ) $$ 
and
$$ \EE \vphi''(f(X)) |\nabla f(X)|^2 \geq (1 - \eps M) \eps^2 \vphi''(1) \cdot \EE |\nabla g(X)|^2. $$
By using the $\vphi$-Sobolev inequality and dividing by $\eps^2 \vphi''(1) > 0$, and then letting $\eps$ tends to zero we obtain
\begin{equation} \frac{1}{2} Var(g(X)) \leq \frac{C_{\vphi}(X)}{2} \EE |\nabla g(X)|^2. \label{eq_1329} \end{equation}
This inequality holds for a Lipschitz, bounded function $g$. A standard truncation argument 
shows that (\ref{eq_1329}) holds true for an arbitrary locally-Lipschitz function $g$. Indeed, It suffices to replace $g(x)$ by $g_M(x) = \theta(x/M) g(x)$ where $\theta: \RR^n \rightarrow [0,1]$ is a compactly supported smooth function that equals one in a neighborhood of the origin. The function $g_M$ is a Lipschitz, bounded function, and a short argument detailed e.g. in the proof of Proposition 27 in \cite{BK} allows to take the limit as $M \rightarrow \infty$ in the Poincar\'e inequality. This shows that $C_P(X) \leq C_{\vphi}(X)$.	
\end{proof}

\begin{proof}[Proof of Proposition \ref{prop_1626}] Let $Z$ be a standard Gaussian random vector, independent of $X$.
When $X$ is a standard Gaussian, the operator $Q_s$ has a pleasant form given by the Mehler formula. Recall that in general, 
$$ Q_s f(X + \sqrt{s} Z) = \EE \left[ f(X) | X + \sqrt{s} Z \right] $$
and hence
$$ \EE Q_s f(X + \sqrt{s} Z) = \EE f(X). $$
In the case where $X$ is Gaussian, the density of $X$ conditioned on $X + \sqrt{s} Z = y$ is
$$ \frac{\gamma_1(x) \gamma_s(y-x)}{\gamma_{1+s}(y)} = \gamma_{\frac{s}{s+1}} \left(x - \frac{y}{s+1} \right) $$
where $\gamma_s$ is the Gaussian density, as in (\ref{eq_1716}). In other words, the distribution of $X$ conditioned on $X + \sqrt{s} Z$ is that of a Gaussian random vector of mean 
$\frac{X + \sqrt{s} Z}{s+1}$ and covariance $\frac{s}{s+1} \id$. Therefore, for any $y \in \RR^n$,
\begin{equation}  Q_s f(y) = \EE f \left(\frac{y}{s+1} + \sqrt\frac{s}{s+1} Z \right) \label{eq_1728} \end{equation}
and 
$$ \nabla Q_s f(y) = \frac{1}{s+1} \EE \nabla f \left(\frac{y}{s+1} + \sqrt\frac{s}{s+1} Z \right). $$
Consequently, in all of $\RR^n$ we have the identity, \begin{equation} \vphi''(Q_s f) |\nabla Q_s f|^2 = \frac{1}{(s+1)^2} \cdot \vphi''(Q_s f) |Q_s (\nabla  f)|^2 \label{eq_1729A} \end{equation}
where $Q_s$ is acting on a  vector field entry by entry, i.e., $Q_s(\nabla f) = (Q_s (\partial_1 f), \ldots Q_s(\partial_n f))$. 
Since the function in (\ref{eq_1547}) is convex in $x$ and $y$, and since $Q_s$ is an averaging operator as we see from (\ref{eq_1728}), we may use Jensen's inequality. 
We obtain from (\ref{eq_1729A}) that
\begin{equation} \vphi''(Q_s f) |\nabla Q_s f|^2 \leq \frac{1}{(s+1)^2} \cdot Q_s \left[ \vphi''(f) |\nabla  f|^2 \right] \label{eq_1729_} \end{equation}
Abbreviate $f_s = Q_s f$ and $X_s = X + \sqrt{s} Z$, and observe that 
$$ \EE \vphi(f(X)) - \vphi(\EE f(X)) = -\int_0^{\infty} \frac{\partial}{\partial s} \EE \vphi(\EE \left[ f(X) | X + \sqrt{s} Z \right] ) ds 
= -\int_0^{\infty} \frac{\partial}{\partial s} \EE \vphi(f_s(X_s))  ds.
$$
From the generalized de Bruijn identity in the form of Lemma \ref{lem_620} above, and from (\ref{eq_1729_}),
\begin{align*}  \EE \vphi(f(X)) - \vphi(\EE f(X)) & = \frac{1}{2} \int_0^{\infty}  \EE \vphi''(f_s(X_s)) |\nabla f_s(X_s)|^2 ds
\\ & \leq   \int_0^{\infty} \frac{1}{2 (s+1)^2} \EE Q_s \vphi''(f(X_s)) |\nabla f(X_s)|^2 ds 
\\ & =  \int_0^{\infty}  \frac{1}{2 (s+1)^2} \EE \left[ \vphi''(f) |\nabla f|^2 \right](X)) ds 
\\ & = \frac{1}{2} \EE \left[ \vphi''(f) |\nabla f|^2 \right](X). 
\end{align*} 
This shows that $C_{\vphi}(X) \leq 1$. The inequality $C_{\vphi}(X) \geq C_P(X) = 1$ follows from 
Proposition \ref{prop_1744}.
\end{proof}

\begin{question}
Is there an extremality property of the log-Sobolev constant $\CLS(X)$ among all $\vphi$-Sobolev constants
$C_{\vphi}(X)$ where $\vphi$ ranges over the class of convex function with $\varphi(1)=0$ such that $1 / \vphi''$ is concave? 
\end{question}

\begin{remark}
In the proof of Proposition \ref{prop_1626} we used the fact that when $X$ is a standard Gaussian random vector,  for any $f$,
\begin{equation}  |\nabla Q_s f|^2 \leq \frac{1}{(s+1)^2} Q_s(|\nabla f|^2). 
	\label{eq_1126} \end{equation}
Since $\vphi'' \geq 0$, inequality (\ref{eq_1126}) and the convexity of the function $(x,y) \mapsto \vphi''(x) y^2$, suffices for concluding (\ref{eq_1729_}).  Inequality (\ref{eq_1126}) is the only property of a Gaussian random vector that was used in the proof of Proposition \ref{prop_1626}.  
Arguing as in the proof of \cite[Lemma 2.5]{KP}, one can show that 
(\ref{eq_1126}) holds true in the case where the law $\mu$ of $X$ is more log-concave than the Gaussian measure, i.e. when its density $e^{-\psi}$ 
satisfies $\nabla^2 \psi \geq \id$ everywhere in $\RR^n$.  We omit the details. This yields another proof of a result from Chafa\"i \cite{chafai}: If $1 / \vphi''$ is concave and $X$ is more log-concave than the Gaussian, then,
	$$ C_{\vphi}(X) \leq 1. $$
\end{remark}

\begin{proof}[Proof of Theorem \ref{thm_1601} and Theorem \ref{thm:phi_convex}] We mimick the proof of Theorem \ref{thm:etachi2UB}.
Begin with the proof of the right-hand side inequality in (\ref{eq_1120_}),
as the left-hand side inequality was already proven above. 
If $C_\vphi(X) = +\infty$ then this inequality is vacuously true, hence we may assume that $C_\vphi(X) < \infty$. 
Assume that $D_{\vphi}( \nu \D \mu) < \infty$,
and write $f = d \nu / d \mu$, so that $f \geq 0$ satisfies $\int f d \mu = 1$. 
Clearly $f_s = Q_s f$ is a probability density with respect to $\mu_s$, and in fact $d \nu_s / d \mu_s = Q_s f$.
 We need to show that 
\begin{equation} D_{\vphi}(\nu_s \D \mu_s) = \int_{\RR^n} \vphi(Q_s f) d \mu_s  \leq \frac{1}{1 + s / C_\vphi(\mu)}  \cdot \int_{\RR^n} \vphi(f) d \mu = \frac{1}{1 + s / C_\vphi(\mu)}  \cdot D_{\vphi}(\nu \D \mu).
\label{eq_543_} \end{equation}
The two equalities in (\ref{eq_543_}) follow from the definition of $D_{\vphi}$,  and it remains to prove the inequality. 
Let $X \sim \mu$ and recall that $X_s = X + \sqrt{s} Z$.
By the generalized de Bruijn identity, i.e. Proposition \ref{thm:deBruijnGeneralized2}, and the definition of the $\vphi$-Sobolev constant,
$$
\frac{\partial}{\partial s} D_{\vphi}(\nu_s \D \mu_s) = -\frac{1}{2} J_{\vphi}(\nu_s \D \mu_s) \leq -\frac{1}{C_\vphi(X_s)} D_{\vphi}(\nu_s \D \mu_s)
\leq -\frac{1}{C_\vphi(X) + C_{\vphi}(\sqrt{s} Z)} D_{\vphi}(\nu_s \D \mu_s),
$$
where in the last passage we used subadditivity, i.e., Theorem \ref{thm_1315}. By homogeneity and Proposition \ref{prop_1626},
$$
\frac{\partial}{\partial s} D_{\vphi}(\nu_s \D \mu_s) \leq -\frac{1}{C_\vphi(X) + s} D_{\vphi}(\nu_s \D \mu_s). 
$$
 Therefore
$$ \frac{\partial}{\partial s} \log D_{\vphi}(\nu_s \D \mu_s) \leq -\frac{1}{s +C_\vphi(X)}.
$$
By integrating from $0$ to $s$ we conclude that 
$$ \log \frac{D_{\vphi}(\nu_s \D \mu_s)}{D_{\vphi}(\nu \D \mu)}   \leq
-\int_0^s \frac{dx}{x + C_\vphi(X)}  = \log \frac{1}{1 + s/C_\vphi(X)}. $$ 
This proves the right-hand side inequality in (\ref{eq_1120_}).

\medskip Let us now assume that $X$ is log-concave and prove (\ref{eq_1603}). Let $0 < \eps < C_\vphi(\mu)$. 
There exists
non-negative function $f$ with $\EE f(X) = 1$ 
such that \begin{equation}  \EE \vphi(f(X))  \geq (C_\vphi(\mu) - \eps) \cdot \EE \vphi''(f(X)) |\nabla f(X)|^2. \label{eq_610_} \end{equation}
By an argument similar to the one from the appendix of \cite{BK}, we may assume that $f$ is smooth.
Let $\nu$ satisfy $d \nu / d \mu = f$.
By Conjecture \ref{lem_1824},
\begin{align}  \nonumber \frac{d}{ds} & J_{\vphi}(\nu_s \D\mu_s) = \frac{d}{ds} \int_{\RR^n} \vphi''(f_s) |\nabla f_s|^2 d \mu_s  \\ & = 
	-\int_{\RR^n}  \left[ 2 (\nabla^2 \psi_s) \nabla g_s \cdot \nabla g_s + \left|\nabla^2 g_s \right|^2 + \kappa(f_s) |\nabla g_s|^4  \right] \frac{d \mu_s}{\vphi''(f_s)} \leq 0, \label{eq_1843}
\end{align} 
where $f_s = Q_s f$. Indeed, each of the three summands in the integral is non-negative; the first is non-negative since $\nabla^2 \psi_s \geq 0$ by log-concavity, the second is always non-negative, and the third is non-negative since $\kappa \geq 0$ as $1 / \vphi''$ is concave. We thus proved that $J_{\vphi}(\nu_s \D \mu_s)$ is decreasing with $s$.
Since $J_{\vphi}(\nu_s \D \mu_s)$ is the derivative of $D_{\vphi}(\nu_s \D\mu_s)$, we conclude that $D_{\vphi}(\nu_s \D\mu_s)$ is a concave function of $s$. Moreover,
by the generalized de Brujin identity,
$$ D_{\vphi}(\nu_s \D \mu_s) - D_{\vphi}(\nu \D\mu)= -\int_0^s J_{\vphi}(\nu_t \D \mu_t) dt \geq -s \cdot J_{\vphi}(\nu \D\mu) \geq -\frac{s}{C_\vphi(\mu) - \eps} \cdot D_{\vphi}(\nu \D \mu). $$
and hence
$$ D_{\vphi}(\nu_s \D \mu_s)  \geq \left( 1-\frac{s}{C_\vphi(\mu) - \eps} \right) D_{\vphi}(\nu \D \mu). $$
Since $\eps > 0$ is arbitrary, this proves that 
$$ \eta_{\vphi}(\mu, s)  \geq 1-\frac{s}{C_\vphi(\mu)}. $$
\end{proof}

\begin{proof}[Proof of Theorem \ref{thm:etaKLbounds} and Theorem \ref{thm:KLconvex}]  
These follow from Theorem \ref{thm_1601} and Theorem \ref{thm:phi_convex}, respectively.
\end{proof}

\begin{appendices}
\section{Technical justifications for Lemma~\ref{lem_620}}

We assume $\vphi(x)=\vphi_{\lambda}(x)=x^{\lambda}-1$, for $\lambda>1$. Let $\nu$ and $\mu$ be two probability distributions on $\RR^n$. We make the following regularity assumptions: 
\begin{enumerate}
\item $D_{\vphi}(\nu\|\mu)<\infty$
\item $\EE_\mu|X|^4<\infty$
\end{enumerate}
Let $f=d\nu/d\mu$, such that $f_s=d\nu_s/d\mu_s$. The following proposition will be useful in the derivation below.
\begin{proposition}
Let $\vphi(x)=\vphi_\lambda(x)=x^{\lambda}-1$, $\lambda>1$, and assume $D_\vphi(\nu\|\mu)<\infty$. Then, for any $s>0$ we also have $D_{\vphi^2}(\nu_s\|\mu_s)<\infty$, and in particular, $\chi^2(\nu\|\mu)<\infty$. Furthermore, $D_{\vphi_{4\lambda}}(\nu_s\|\mu_s)<\infty$.
\label{prop:ReniyFinite}
\end{proposition}

\begin{proof}
Recall that the R\`enyi entropy of order $\lambda$ is defined as
\begin{align}
D_{\lambda}(\nu_s\|\mu_s)=\frac{1}{\lambda-1}\log \left(1+\EE\left[\vphi_{\lambda}\left(\frac{d\nu_s}{d\mu_s}(X_s)\right)\right]\right)=\frac{1}{\lambda-1}\log\left(1+ \EE\left[\vphi_{\lambda}\left(f_s(X_s)\right)\right]\right),\label{eq:RenyiDef}
\end{align}
where $\vphi_{\lambda}(x)=x^\lambda-1$, and $f=d\nu/d\mu$, such that $f_s=d\nu_s/d\mu_s$. It therefore follows that
\begin{align}
\vphi^2_{\lambda}(x)=(x^{\lambda}-1)^2=(x^{2\lambda}-1)-2(x^{\lambda}-1)=\vphi_{2\lambda}(x)-2\vphi_{\lambda}(x),
\end{align}
so that
\begin{align}
D_\vphi(\nu_s\|\mu_s)&=\EE \vphi_{\lambda}(f_s(X_s))=e^{(\lambda-1)D_{\lambda}(\nu_s\|\mu_s)}-1\label{eq:RenyiExpt}\\
D_{\vphi^2}(\nu_s\|\mu_s)&=\EE \vphi_{2\lambda}(f_s(X_s))-2\EE \vphi_{\lambda}(f_s(X_s))=e^{(2\lambda-1)D_{2\lambda}(\nu_s\|\mu_s)}-2e^{(\lambda-1)D_{\lambda}(\nu_s\|\mu_s)}+1.
\end{align}
By the data-processing inequality, and the assumption that $D_{\vphi}(\nu\|\mu)$ is finite, we have
\begin{align}
D_\vphi(\nu_s\|\mu_s)\leq D_\vphi(\nu\|\mu)<\infty.
\end{align}
Therefore, by~\eqref{eq:RenyiDef}, it follows that $D_{\lambda}(\nu_s\|\mu_s)$ is also finite.
From~\cite[Theorem 1]{anantharam2018variational} (see also~\cite{shayevitz2011renyi} and \cite[Theorem 30]{ErvenHarremos14}) it follows that finiteness of $D_{\lambda}(\nu_s\|\mu_s)$ for $\lambda>1$ also implies the finiteness of $D_{2\lambda}(\nu_s\|\mu_s)$ (and also of $D_{4\lambda}(\nu_s\|\mu_s)$). To see this, note that for $\lambda>1$~\cite[Theorem 1]{anantharam2018variational} reads
$$D_{\lambda}(\nu_s\|\mu_s)=\sup_{\xi\ll \nu_s}\left[D_{\text{KL}}(\xi\|\mu_s)-\frac{\lambda}{\lambda-1}D_{\text{KL}}(\xi\|\nu_s)\right].$$
Thus, finiteness of $D_{\lambda}(\nu_s\|\mu_s)$ implies that for any distribution $\xi\ll \nu_s$ we have $D_{\text{KL}}(\xi\|\mu_s)<\infty$. Using the same equation again, with $2\lambda$ (or $4\lambda)$, we see that  $D_{2\lambda}(\nu_s\|\mu_s)$ must be finite as well. Thus 
\begin{align}
D_\vphi(\nu_s\|\mu_s)<\infty \Rightarrow D_{\lambda}(\nu_s\|\mu_s) < \infty \Rightarrow D_{2\lambda}(\nu_s\|\mu_s)<\infty \Rightarrow D_{\vphi^2}(\nu_s\|\mu_s)<\infty.
\end{align}
To prove that $\chi^2(\nu\|\mu)$ is finite, we note that $\chi^2(\nu\|\mu)=D_{\vphi_{\lambda'}}(\nu\|\mu)$ with $\lambda'=2$. Since $\lambda>1$, we have that $\lambda'<2\lambda$. We have already shown that $D_{2\lambda}(\nu\|\mu)$ is finite. This together with the fact that $\lambda\mapsto D_{\lambda}(\nu\|\mu)$ is non-decreasing in $\lambda>1$, shows that $D_{\lambda'}(\nu\|\mu)$ is finite, which in turn implies that $D_{\vphi_{\lambda'}}(\nu\|\mu)$ is finite.
\end{proof}

\subsection{Justification of first integration in parts}
\label{subsec:firstintbyparts}
We prove that
\begin{align}
\sum_{i=1}^n\int_{\RR^n}|\vphi(f_s)\frac{\partial}{\partial x_i} \rho_s|<\infty.
\label{eq:absintegrable1}
\end{align}
We have 
\begin{align}
\vphi(f_s)\nabla \rho_s=\vphi(f_s)\sqrt{\rho_s}\cdot \frac{\nabla \rho_s}{\sqrt{\rho_s}}.
\end{align}
Applying the Cauchy-Schwartz inequality, we therefore have that
\begin{align}
\sum_{i=1}^n\int_{\RR^n}|\vphi(f_s)\frac{\partial}{\partial x_i} \rho_s| &=n\sum_{i=1}^n\frac{1}{n}\int_{\RR^n}|\vphi(f_s)\sqrt{\rho_s}|\cdot|\frac{\frac{\partial}{\partial x_i} \rho_s}{\sqrt{\rho_s}}|\\
&\leq n\sum_{i=1}^n\frac{1}{n}\sqrt{\int_{\RR^n}\vphi^2(f_s) \rho_s \cdot \int_{\RR^n} \frac{\left(\frac{\partial}{\partial x_i} \rho_s  
\right)^2}{\sqrt{\rho_s}} }\\
&\leq \sqrt{n}\sqrt{\int_{\RR^n}\vphi^2(f_s) \rho_s \cdot \int_{\RR^n} \frac{|\nabla \rho_s|^2}{\rho_s} }\\
&=\sqrt{\EE \vphi^2(f_s(X_s))}\sqrt{n\cdot \m{J}(\mu_s)}.
\end{align}
We have that $\m{J}(\mu_s)<\infty$ for any $s>0$. Finiteness of $\EE \vphi^2(f_s(X_s))=D_{\vphi^2}(\nu_s\|\mu_s)$ follows from Proposition~\ref{prop:ReniyFinite}. We therefore conclude that~\eqref{eq:absintegrable1} holds.


\subsection{Justification of second integration in parts}

We prove that
\begin{align}
\sum_{i=1}^n\int_{\RR^n}|\rho_s\frac{\partial}{\partial x_i} \vphi(f_s)|<\infty.
\label{eq:absintegrable2}
\end{align}
Recalling that $\vphi(x)=x^{\lambda}-1$ we have
\begin{align}
\vphi'(x)&=\lambda x^{\lambda-1},~~\vphi''(x)=\lambda(\lambda-1)x^{\lambda-2},
\end{align}
and therefore
\begin{align}
\frac{\vphi'(x)}{\sqrt{\vphi''(x)}}&=\left(\frac{\lambda}{\lambda-1}x^{\lambda} \right)^{1/2}=\left(\frac{\lambda}{\lambda-1}(\vphi(x)+1) \right)^{1/2}.
\end{align}
Consequently, (recall that $\vphi''>0$)
\begin{align}
\rho_s\nabla \vphi(f_s)&=\rho_s\vphi'(f_s)\nabla f_s=\left(\sqrt{\vphi''(f_s)\rho_s}\nabla f_s\right)\left(\frac{\vphi'(f_s)}{\sqrt{\vphi''(f_s)}}\sqrt{\rho_s} \right)\\
&=\sqrt{\frac{\lambda}{\lambda-1}}\left(\sqrt{\vphi''(f_s)\rho_s}\nabla f_s\right)\left(\sqrt{(\vphi(f_s)+1)\rho_s} \right).
\end{align}
Applying the Cauchy-Schwartz inequality, we therefore have that
\begin{align}
\sum_{i=1}^n\int_{\RR^n}|\rho_s\frac{\partial}{\partial x_i} \vphi(f_s)|&=\sqrt{\frac{\lambda}{\lambda-1}}n\sum_{i=1}^n\frac{1}{n}\int_{\RR^n} \left(\sqrt{\vphi''(f_s)\rho_s}\frac{\partial}{\partial x_i} f_s\right)\left(\sqrt{(\vphi(f_s)+1)\rho_s} \right)\\
&\leq \sqrt{\frac{\lambda}{\lambda-1}}n\sum_{i=1}^n\frac{1}{n}\sqrt{\int_{\RR^n}\vphi''(f_s)\rho_s\left(\frac{\partial}{\partial x_i} f_s\right)^2\int_{\RR^n}(\vphi(f_s)+1)\rho_s}\\
&=\sqrt{\frac{\lambda}{\lambda-1}}\sqrt{n(1+\EE \vphi(f_s(X_s)))\cdot \int_{\RR^n}\vphi''(f_s)|\nabla f_s|^2d\mu_s}.
\end{align}
Consequently, the assumption that $D_{\vphi}(\nu_s\|\mu_s)$ is finite and that $\int_{\RR^n}\vphi''(f_s)|\nabla f_s|^2d\mu_s$ is finite (otherwise the statement is void), implies that~\eqref{eq:absintegrable2} holds.

\subsection{Justification of taking derivative under the integral sign}

We show that $\frac{d}{ds}\int_{\RR^n}  \vphi(f_s) \rho_s=\int_{\RR^n} \frac{d}{ds} \vphi(f_s) \rho_s$. To that end, we will show that for any $0<a<b$ close enough to each other, it holds that $\sup_{s\in[a,b]} |\frac{d}{ds} \vphi(f_s) \rho_s|$ is integrable, and  use the dominated convergence theorem. In particular, we may assume $\frac{b}{a}\leq \frac{4\lambda-1}{4\lambda-2}$.

Let $\beta$ be the density corresponding to $\nu$, and $\beta_s=\beta*\gamma_s$ the density corresponding to $\nu_s$. Recalling that $\vphi(x)=x^\lambda-1$, we have that $\vphi(f_s)=\left(\frac{\beta_s}{\rho_s} \right)^{\lambda}-1$, so that
\begin{align}
\frac{d}{ds} \vphi(f_s) \rho_s  &= \frac{d}{ds}\left(\beta_s^{\lambda}\rho_s^{1-\lambda}-\rho_s \right)\\
&=\lambda \beta'_s\left(\frac{\beta_s}{\rho_s} \right)^{\lambda-1}+(1-\lambda)\rho'_s\left(\frac{\beta_s}{\rho_s} \right)^{\lambda}-\rho'_s\\
&=\rho_s\left(\frac{\beta_s}{\rho_s} \right)^{\lambda}\left(\lambda\cdot\frac{\beta'_s}{\beta_s}+(1-\lambda)\cdot\frac{\rho'_s}{\rho_s} \right)-\rho'_s.
\end{align}
Thus,
\begin{align}
&\int_{\RR^n}\sup_{s\in[a,b]}|\frac{d}{ds} \vphi(f_s) \rho_s|\\
&\leq \lambda \underbrace{\int_{\RR^n}\sup_{s\in[a,b]} \left|\rho_s\left(\frac{\beta_s}{\rho_s} \right)^{\lambda}\frac{\beta'_s}{\beta_s} \right|}_{I}+|1-\lambda|\underbrace{\int_{\RR^n}\sup_{s\in[a,b]} \left|\rho_s\left(\frac{\beta_s}{\rho_s} \right)^{\lambda}\frac{\rho'_s}{\rho_s} \right|}_{II}+\underbrace{\int_{\RR^n}\sup_{s\in[a,b]}|\rho'_s|}_{III}.
\end{align}
Before bounding integrals $I$, $II$, and $III$, we develop explicit expressions for $\rho_s$ ,$\beta_s$, $\rho'_s$ and $\beta'_s$. Recall that
\begin{align}
\rho_s(y)&=\EE_\mu [\gamma_s(y-X)],\\
\frac{d}{ds}\gamma_s(y-X)&=\frac{1}{s}\gamma_s(y-X)\left[\frac{|y-X|^2}{2s}-\frac{n}{2}\right].
\end{align}
Differentiation under the integral sign is valid, see e.g.~\cite{barron1984monotonic}, and therefore
\begin{align}
\rho'_s(y)=\frac{d}{ds}\rho_s(y)&=\EE_\mu\left[\frac{d}{ds}\gamma_s(y-X)\right]=\frac{1}{s}\left[\EE_{\mu}\left[\frac{|y-X|^2}{2s}\gamma_s(y-X) \right]-\frac{n}{2}\rho_s(y) \right].\label{eq:drho}
\end{align}
Similarly,
\begin{align}
\beta_s(y)=\EE_\nu [\gamma_s(y-X)],~~~  \beta'_s(y)=\frac{d}{ds}\beta_s(y)=\frac{1}{s}\left[\EE_\nu\left[\frac{|y-X|^2}{2s}\gamma_s(y-X) \right]-\frac{n}{2}\beta_s(y) \right].\label{eq:dbeta}  
\end{align}

We begin with showing that the integral $I$ is finite. Using the Cauchy-Schwartz inequality, we have
\begin{align}
\int_{\RR^n}\sup_{s\in[a,b]} \left|\rho_s\left(\frac{\beta_s}{\rho_s} \right)^{\lambda}\frac{\beta'_s}{\beta_s} \right|\leq \left(\int_{\RR^n}\sup_{s\in[a,b]} \left|\rho_s\left(\frac{\beta_s}{\rho_s} \right)^{2\lambda}\right|\cdot \int_{\RR^n}\sup_{s\in[a,b]} \left|\rho_s\left(\frac{\beta'_s}{\beta_s}\right)^2 \right| \right)^{1/2}.\label{eq:integralIcs}    
\end{align}
We bound each of the two integrals in the product above. We have that
\begin{align}
&\sup_{s\in[a,b]} \left|\rho_s\left(\frac{\beta_s}{\rho_s} \right)^{2\lambda}(y)\right|\leq \sup_{s\in[a,b]} (2\pi s)^{-n/2} \sup_{s\in[a,b]} \bigg(\EE_\nu \left[\exp\left(-\frac{1}{2s}|y-X|^2\bigg)\right]\right)^{2\lambda}\\
&~~~~~~~~~~~~~~~~~~~~~~~~~~~~~~~~~~~~~\cdot\sup_{s\in[a,b]} \left(\EE_\mu \left[\exp\left(-\frac{1}{2s}|y-X|^2\right)\right]\right)^{1-2\lambda}\\
&=(2\pi a)^{-n/2}\bigg(\EE_\nu \left[\exp\left(-\frac{1}{2b}|y-X|^2\bigg)\right]\right)^{2\lambda} \left(\EE_\mu \left[\exp\left(-\frac{1}{2a}|y-X|^2\right)\right]\right)^{1-2\lambda}\\
&=C \rho_a^{1-2\lambda}\beta_b^{2\lambda}=C\rho_a \left(\frac{\beta_b}{\rho_a} \right)^{2\lambda}=C\rho_a \left(\frac{\rho_b}{\rho_a} \frac{\beta_b}{\rho_b} \right)^{2\lambda}\\
&=C\rho_a \left[\left(\frac{\rho_b}{\rho_a}  \right)^{2\lambda-1/2}\right]\left[\left(\frac{\rho_b}{\rho_a}\right)^{1/2}\left( \frac{\beta_b}{\rho_b} \right)^{2\lambda}\right]
\end{align}
where $C=(b/a)^{n\lambda}$. With this, we may apply the Cauchy-Schwartz inequality and obtain
\begin{align}
\int_{\RR^n}\sup_{s\in[a,b]} \left|\rho_s\left(\frac{\beta_s}{\rho_s} \right)^{2\lambda}\right|\leq C \left(\EE_{\mu_a}\left[\left(\frac{\rho_b}{\rho_a} \right)^{4\lambda-1}\right]\EE_{\mu_b}\left[\left(\frac{\beta_b}{\rho_b} \right)^{4\lambda}\right]\right)^{1/2}.
\end{align}
We have that $\EE_{\mu_b}\left[\left(\frac{\beta_b}{\rho_b} \right)^{4\lambda}\right]=1+D_{\vphi_{4\lambda}}(\nu_b\|\mu_b)$ is finite due to Proposition~\ref{prop:ReniyFinite}. Furthermore, by the data-processing inequality
\begin{align}
 \EE_{\mu_a}\left[\left(\frac{\rho_b}{\rho_a} \right)^{4\lambda-1}\right]=1+D_{\vphi_{4\lambda-1}}(\mu_b\|\mu_a)\leq 1+D_{\vphi_{4\lambda-1}}(\m{N}(0,b\cdot \id )\|\m{N}(0,a\cdot \id )).   
\end{align}
It is easy to see that $D_{\vphi_{\lambda}}(\m{N}(0,\sigma^2_0\cdot \id )\|\m{N}(0,\sigma^2_1\cdot \id )<\infty$ provided that $\frac{\sigma_0^2}{\sigma_1^2}<\frac{\lambda}{\lambda-1}$, see~\cite[eq. (10)]{ErvenHarremos14}. Thus, $\EE_{\mu_a}\left[\left(\frac{\rho_b}{\rho_a} \right)^{4\lambda-1}\right]<\infty$, since we assumed $\frac{b}{a}\leq \frac{4\lambda-1}{4\lambda-2}$.

We move on to upper bounding the second integral in~\eqref{eq:integralIcs}. Using~\eqref{eq:dbeta} we have that
\begin{align}
\left|\frac{\beta'_s}{\beta_s}\right|=\frac{1}{s}\left|\frac{\EE_\nu\left[\frac{|y-X|^2}{2s}\gamma_s(y-X) \right]-\frac{n}{2}\beta_s}{\beta_s}\right|\leq \frac{1}{s}\left[\frac{n}{2}+\frac{\EE_\nu\left[\frac{|y-X|^2}{2s}\gamma_s(y-X) \right]}{\beta_s} \right]
\end{align}
The function $|y-X|\mapsto \frac{|y-X|^2}{2s}$ is increasing, while the function $|y-X|\mapsto \gamma_s(|y-X|)$ is decreasing. Consequently,
\begin{align}
\EE_\nu\left[\frac{|y-X|^2}{2s}\gamma_s(y-X) \right]\leq \EE_\nu\left[\frac{|y-X|^2}{2s} \right] \EE_\nu\left[\gamma_s(y-X) \right]=\beta_s\cdot \EE_\nu\left[\frac{|y-X|^2}{2s} \right].
\end{align}
Using the Cauchy Schwartz inequality, we can further bound
\begin{align}
\EE_\nu\left[\frac{|y-X|^2}{2s} \right]&=\EE_\mu\left[\frac{d\nu}{d\mu}(X)\frac{|y-X|^2}{2s} \right]\leq \left(\EE_\mu\left[\left(\frac{d\nu}{d\mu}\right)^2\right]\cdot \EE_\mu\left[\frac{|y-X|^4}{4s^2} \right]\right)^{1/2}\label{eq:changeofmeasurestart}\\
&=\left((\chi^2(\nu\|\mu)+1)\cdot \EE_\mu\left[\frac{|y-X|^4}{4s^2} \right]\right)^{1/2}\\
&\leq c \cdot \left( \EE_\mu\left[\frac{|y-X|^4}{4s^2} \right]\right)^{1/2},\label{eq:changeofmeasuresend}
\end{align}
where the fact that $\chi^2(\nu\|\mu)$ is bounded follows from Proposition~\ref{prop:ReniyFinite}. Thus,
\begin{align}
\left|\frac{\beta'_s}{\beta_s}\right|\leq \frac{1}{s}\left[\frac{n}{2}+c \cdot \left( \EE_\mu\left[\frac{|y-X|^4}{4s^2} \right]\right)^{1/2}\right].
\end{align}
In particular, for $s\in[a,b]$, we have
\begin{align}
\left|\frac{\beta'_s}{\beta_s}\right|^2&\leq \frac{1}{a^2}\left[\frac{n^2}{4}+cn\left( \EE_\mu\left[\frac{|y-X|^4}{4a^2} \right]\right)^{1/2} +c^2\cdot\EE_\mu\left[\frac{|y-X|^4}{4a^2} \right]\right]\\
&=c_1+c_2\left( \EE_\mu\left[\frac{|y-X|^4}{4b^2} \right]\right)^{1/2}+c_3 \EE_\mu\left[\frac{|y-X|^4}{4b^2} \right].
\end{align}
Noting further that for all $s\in[a,b]$ we have that $\rho_s\leq (b/a)^{n/2}\rho_b$, it holds that
\begin{align}
\int_{\RR^n}\sup_{s\in[a,b]} &\left|\rho_s\left(\frac{\beta'_s}{\beta_s}\right)^2\right|\leq \int_{\RR^n}c'_1 \rho_b(y)+c'_2\rho_b(y)\left( \EE_\mu\left[\frac{|y-X|^4}{4b^2} \right]\right)^{1/2}+c'_3\rho_b(y) \EE_\mu\left[\frac{|y-X|^4}{4b^2} \right]dy\\
&\leq c'_1+c'_2 \left( \EE_{Y\sim\rho_b}\EE_{X\sim\mu}\left[\frac{|Y-X|^4}{4b^2} \right]\right)^{1/2}+c'_3 \EE_{Y\sim\rho_b}\EE_{X\sim\mu}\left[\frac{|Y-X|^4}{4b^2} \right],
\end{align}
where we have used Jensen's inequality above. The expression above is finite by our assumption that $\EE_\mu|X|^4<\infty$. Thus, integral $I$ is finite. 

The proof that integral $II$ is finite is nearly identical, where the only difference is that we do not need the change of measure from $\nu$ to $\mu$ ``trick'' we have used in~\eqref{eq:changeofmeasurestart}-\eqref{eq:changeofmeasuresend}, and instead we have the trivial bound $\EE_\mu\left[\frac{|y-X|^2}{2s} \right]\leq \left(\EE_\mu\left[\frac{|y-X|^4}{4s^2} \right] \right)^{1/2}$.

We are left with showing that integral $III$ converges. This follows since
\begin{align}
\int_{\RR^n}\sup_{s\in[a,b]}|&\rho'_s|=\int_{\RR^n}\sup_{s\in[a,b]}\rho_s\left|\frac{\rho'_s}{\rho_s}\right|\leq \left(\int_{\RR^n}\sup_{s\in[a,b]}\rho_s\cdot \int_{\RR^n}\sup_{s\in[a,b]}\rho_s \left|\frac{\rho'_s}{\rho_s}\right|^2\right)^{1/2}\\
&\leq (b/a)^{n/2}\left( \int_{\RR^n}\sup_{s\in[a,b]}\rho_s \left|\frac{\rho'_s}{\rho_s}\right|^2\right)^{1/2},
\end{align}
and we have already shown that this integral converges.

\subsection{Proof of that~\eqref{eq:generalderivative} implies~\eqref{eq_1713}}

Let $\lambda> 1$ and $\vphi_\lambda(x)=x^{\lambda}-1$. By~\eqref{eq:generalderivative}, we have
that for any $s > 0$,
\begin{align}
\frac{d}{ds} D_{\vphi_{\lambda}}( \nu_s \D \mu_s) = -\frac{1}{2} J_{\vphi_\lambda}( \nu_s \D \mu_s).
\end{align}
where 
$$ J_{\vphi_{\lambda}}(\nu \D \mu) = \int_{\RR^n} \vphi_{\lambda}''(f) |\nabla f|^2 d \mu=\lambda(\lambda-1)\int_{\RR^n} f^{\lambda-1}\frac{ |\nabla f|^2}{f} d \mu, $$
where $f=d\nu/d\mu$. Recalling the definition of R\'enyi divergence~\eqref{eq:RenyiDef}, we have that for any $\lambda>1$
\begin{align}
\frac{d}{ds}D_{\lambda}(\nu_s\|\mu_s)=-\frac{1}{2}\frac{\lambda\int_{\RR^n} f_s^{\lambda-1}\frac{ |\nabla f_s|^2}{f_s} d \mu_s}{1+D_{\vphi_\lambda}(\nu_s\|\mu_s)}\triangleq G(\lambda,s).
\end{align}
We have that $G(\lambda,s)$ is continuous in both $\lambda>1$ and $s>0$, and that $\lim_{\lambda\to 1}G(\lambda,s)=-\frac{1}{2}J(\nu_s\|\mu_s)$. Recall that
\begin{align}
D_{\text{KL}}(\nu_s\|\mu_s)=\lim_{\lambda\to 1}D_{\lambda}(\nu_s\|\mu_s).
\end{align}
The theorem then follows by exchanging of limits in
\begin{align}
\frac{d}{ds}\lim_{\lambda\to 1}D_{\lambda}(\nu_s\|\mu_s)=\lim_{\lambda\to 1}\frac{d}{ds} D_{\lambda}(\nu_s\|\mu_s)=\lim_{\lambda\to 1}G(\lambda,s)=-\frac{1}{2}J(\nu_s\|\mu_s),
\end{align}
which is valid because $G(\lambda, s)$ is continuous in $(\lambda, s) \in [1, \infty) \times (0, \infty)$. 
\end{appendices}




\bibliographystyle{plain}
\bibliography{GaussianSDPI_Bib}

\end{document}